\documentclass[onecolumn, 
superscriptaddress,
notitlepage,
longbibliography,
showkeys,
10pt]{revtex4-2}

\usepackage{graphicx}
\usepackage{tikz}
\usepackage{amsmath}
\usepackage{verbatim}
\usepackage{amsfonts}
\usepackage{amssymb}
\usepackage{mathrsfs}
\usepackage{amsthm}
\usepackage{subcaption}
\usepackage{hyperref}
\usepackage{algorithm}
\usepackage{enumerate}
\usepackage{color}
\allowdisplaybreaks
\usepackage{mathtools}
\usepackage{hyperref}
\usepackage{times}
\usepackage[left=1in,right=1in,top=1in,bottom=1in]{geometry}

\usepackage{csquotes}

\newtheorem{theorem}{Theorem}
\newtheorem{corollary}[theorem]{Corollary}

\newtheorem{lemma}[theorem]{Lemma}

\newtheorem{proposition}[theorem]{Proposition}
    
\theoremstyle{remark}
\newtheorem*{remark}{Remark}

\newcommand\myeq[1]{\stackrel{\normalfont{#1}}{=}}

\newcommand{\Real}{\mathbb{R}}
\newcommand{\Complex}{\mathbb{C}}
\newcommand{\ee}{\mathbb{E}}

\newcommand{\unit}{\mathbb{I}}
\newcommand{\mani}{\mathcal{M}}

\newcommand{\secref}[1]{Section~\ref{#1}}
\newcommand{\phase}{P}
\newcommand{\norm}[1]{\lVert#1\rVert}
\newcommand{\comm}[2]{[#1,#2]}
\newcommand{\abs}[1]{\lvert#1\rvert}
\newcommand{\absbig}[1]{\big\lvert#1\big\rvert}
\newcommand{\absBig}[1]{\Big\lvert#1\Big\rvert}
\newcommand{\tr}{{\normalfont\text{tr}}}
\newcommand{\var}{{\normalfont\text{Var}}}

\newcommand{\numsample}{N_s}
\newcommand{\numtime}{N_t}

\usepackage{hyperref}

\hypersetup{
    colorlinks=true,
    linkcolor=blue!80!black!80!,
    citecolor=blue!70!black!80!,
    urlcolor=blue!70!black!80!,
}

\usepackage{todonotes}

\newcommand{\obs}{O}
\newcommand{\loss}{\mathscr{L}}
\newcommand{\ud}{\mathrm{d}}
\newcommand{\NP}{\mathcal{J}}
\newcommand{\dt}{\Delta t}

\newcommand{\la}{\mathopen{\langle}}
\newcommand{\ra}{\mathclose{\rangle}}
\newcommand{\bra}[1]{\la #1\lvert}
\newcommand{\ket}[1]{\rvert #1\ra}

\newcommand{\innerbig}[2]{\big\langle #1,#2\big\rangle}

\newcommand{\inner}[2]{\la #1,#2\ra}
\newcommand{\inneravg}[2]{\la #1 \ra_{#2}}
\newcommand{\avgbraket}[2]{\inner{#2}{#1#2}}

\newcommand{\re}[1]{\mathfrak{Re}\{#1\}}
\newcommand{\rebig}[1]{\mathfrak{Re}\big\{#1\big\}}
\newcommand{\reBig}[1]{\mathfrak{Re}\Big\{#1\Big\}}

\newcommand{\im}[1]{\mathfrak{Im}\{#1\}}

\newcommand{\cstA}{\mathsf{G}}
\newcommand{\CA}{\mathsf{A}}
\newcommand{\CB}{\mathsf{B}}
\newcommand{\CC}{\mathsf{C}}
\newcommand{\CF}{\mathsf{F}}
\newcommand{\CL}{\mathsf{Y}}
\newcommand{\jumprate}{\lambda}
\newcommand{\figwidth}{1.0}
\graphicspath{{assets}}

\newcommand{\appref}[1]{Appendix~\ref{#1}}

\begin{document}

\title{Dynamically Optimal Unraveling Schemes for Simulating Lindblad Equations}

\author{Yu Cao}
\email{yucao@sjtu.edu.cn}
\affiliation{Institute of Natural Sciences, Shanghai Jiao Tong University, Shanghai 200240, China}
\affiliation{School of Mathematical Sciences, Shanghai Jiao Tong University, Shanghai, 200240, China}
\affiliation{Ministry of Education Key Laboratory in Scientific and Engineering Computing, Shanghai Jiao Tong University, Shanghai 200240, China}

\author{Mingfeng He}
\affiliation{Zhiyuan College, Shanghai Jiao Tong University, Shanghai, 200240, China}

\author{Xiantao Li}
\affiliation{Department of Mathematics, Pennsylvania State University, University Park, PA 16802, USA}

\date{\today}

\begin{abstract}
Stochastic unraveling schemes are powerful computational tools for simulating Lindblad equations, offering significant reductions in memory requirements. However, this advantage is accompanied by increased stochastic uncertainty, and the question of optimal unraveling remains open. In this work, we investigate unraveling schemes driven by Brownian motion or Poisson processes and present a comprehensive parametric characterization of these approaches. For the case of a single Lindblad operator and one noise term, this parametric family provides a complete description for unraveling scheme with pathwise norm-preservation. We further analytically derive dynamically optimal quantum state diffusion (DO-QSD) and dynamically optimal quantum jump process (DO-QJP) that minimize the growth rate of the variance of an observable locally in time. Compared to jump process ansatz, DO-QSD offers two notable advantages: firstly, the variance for DO-QSD can be rigorously shown not to exceed that of any jump-process ansatz locally in time; secondly, it has very simple expressions. Numerical results demonstrate that the proposed DO-QSD scheme may achieve substantial reductions in the variance of observables and the resulting simulation error.
\end{abstract}


\maketitle

\section{Introduction}
The inevitable interactions of quantum systems  with their environment trigger distinctive quantum phenomena. A challenging task is to predict the system’s evolution without simulating the complete quantum configuration.  The Lindblad equation is a well-known framework that provides a complete description of completely-positive trace-preserving (CPTP) quantum semigroups \cite{lindblad_generators_1976,gorini_completely_1976}. Since its introduction, the Lindblad equation has been extensively studied across multiple disciplines, including quantum thermodynamics \cite{alicki_introduction_2018,spohn_entropy_1978,KastoryanoTemme2013}, quantum optics \cite{briegel_quantum_1993,carmichael_open_1993,breuer_theory_2007}, and quantum computing \cite{verstraete_quantum_2009}, among others. The general form of the Lindblad equation is:
\begin{equation}\label{eqn::lb}
\begin{aligned}
\dot{\rho}(t) &= \mathcal{L}\big(\rho(t)\big) = -i \comm{H}{\rho(t)} + \sum_{k=1}^K \left( L_k \rho(t) L_k^\dagger - \frac{1}{2} \big\{L_k^\dagger L_k, \rho(t)\big\} \right).
\end{aligned}
\end{equation}
Here, $\mathcal{L}$ is the generator of the Lindblad equation; $H$ is the effective Hamiltonian (including possible Lamb shift); $L_k$ describes the interaction between the system and the environment, and $K$ is the number of Lindblad operators; and $\rho(t)$ is the density matrix at time $t$. For an $n$-dimensional quantum system, $\rho(t)$ is an $n \times n$ positive-semidefinite matrix with trace one.

Due to its mathematical significance and wide range of applications, a fundamental problem is how to efficiently simulate the Lindblad equation. For finite-dimensional systems, the Lindblad equation becomes a matrix-valued ordinary differential equation (ODE), to which standard numerical schemes for ODEs can be applied. However, these solvers typically do not preserve the CPTP properties of the Lindblad equation. Recently, a family of structure-preserving schemes was proposed to simulate the Lindblad equation while maintaining positivity and trace \cite{cao_structure-preserving_2025}. When the dimension $n$ is large, such as in many-body quantum systems (where $n = 2^d$ and $d$ is the number of qubits), storing the quantum system on a classical computer becomes a significant challenge. An efficient approach is to use low-rank approximations of the density matrix. For general matrix-valued ODEs, dynamical low-rank approximation was studied by Koch and Lubich \cite{koch_dynamical_2007}. For the Lindblad equation specifically, a dynamically optimal low-rank approximation was investigated by Le Bris and Rouchon \cite{le_bris_low-rank_2013}, and the low-rankness was also pursued and studied by Appel\"o and Cheng in \cite{appelo2025kraus}.

When the low-rank approximation is not efficient (i.e., the rank of $\rho(t)$ is not very small), memory requirements may still pose a significant challenge. A widely used alternative is to stochastically unravel the Lindblad equation, that is, to find a stochastic Schr{\"o}dinger-type equation whose ensemble average recovers the original density matrix. This approach is closely related to the origin of the density matrix, which was introduced to describe the quantum behavior of a classical ensemble of pure states, often referred to as quantum trajectories.  There are two natural types of stochastic processes for this purpose: (a) Itô stochastic differential equations driven by Brownian motion, and (b) jump processes driven by Poisson processes. In the early 1990s, these two approaches were independently developed: Gisin and Percival proposed quantum state diffusion (QSD) \cite{gisin_quantum-state_1992}, while Dalibard, Castin, and Mølmer introduced the Monte Carlo wave-function (MCWF) method \cite{dalibard_wave-function_1992}, which will also be called the quantum jump process later in this paper. Since these seminal works, stochastic unraveling has become a standard tool for simulating high-dimensional Lindblad equations when memory is a limiting factor. Moreover, stochastic unraveling can be combined with low-rank methods; see, for example, \cite{le_bris_adaptive_2015, cao_stochastic_2018}.

Although the conventional QSD in \cite{gisin_quantum-state_1992} and the quantum jump process in \cite{dalibard_wave-function_1992} both lead to Lindblad equation \eqref{eqn::lb}, some opening questions still remain: is there a much larger class of unraveling dynamics that reproduces the same Lindblad equation? Which stochastic unraveling yields a better performance in terms of the statistical error?
To address these questions, we make the following contributions in this work:

\begin{itemize}

\item {\bf (Complete characterization):} For both diffusion-type and jump-type stochastic processes, we provide a comprehensive characterization of all possible such processes when there is only one Lindblad operator ($K = 1$) and one noise term. For the diffusion case, this problem was nearly solved in the original paper \cite{gisin_quantum-state_1992}, and the final structure that we obtain is similar to the conventional QSD in \cite{gisin_quantum-state_1992}; through mathematical derivation, we note that a phase factor in the diffusion coefficient (which also arise from the quadrature in homodyne measurement) can potentially significantly enhance the performance of stochastic unraveling in terms of variance reduction. See Theorem~\ref{lem::characterization} and Theorem~\ref{thm::qjp_characterize} for details.

\item {\bf (Explicitly solvable dynamically optimal unraveling):} For both diffusion-type and jump-type stochastic processes, we explicitly solve for the dynamically optimal stochastic unraveling by minimizing the rate of change of the variance locally in time. This approach is similar to dynamical low-rank approximation \cite{koch_dynamical_2007,le_bris_low-rank_2013, cao_stochastic_2018}, but instead of optimizing over the space of trajectories, we focus on the variance of a classical observable. See Theorem~\ref{thm::optimal_soln} and Theorem~\ref{thm::optimal_soln_jump} for details. In this work, we refer to the dynamically optimal stochastic processes with Brownian motion as dynamically optimal quantum state diffusion (DO-QSD), and the one with Poisson process as dynamically optimal quantum jump process (DO-QJP).

\item {\bf (Comparison of dynamically optimal unraveling schemes):} Notably, we observe that, with proper optimization, the variance of the dynamically optimal quantum state diffusion never exceeds that of any jump process locally in time. Details are provided in Theorem~\ref{thm::optimal_soln_jump}. Additionally, DO-QSD has a simpler expression compared to DO-QJP in the current form, offering both simplicity and superior performance in general. We need to remark that this discussion has not considered the numerical discretization error, as well as the need for physical interpretability of the stochastic process. Our comparison herein is purely based on the computational perspective, assuming that an accurate numerical solver is available.

\item {\bf (Numerical demonstration):} Numerically, we demonstrate that DO-QSD empirically achieves lower variance and consequently smaller error (when the time step is sufficiently small) compared to conventional QSD. Although DO-QSD is only \emph{locally optimal}, it often matches or outperforms schemes that optimize tunable parametric functions using machine learning. While it remains open whether DO-QSD minimizes variance at a fixed terminal time, its empirical performance indicates that DO-QSD is promising and efficient.

\end{itemize}

We remark that the conclusions in this work can be readily generalized to time-dependent Lindblad equations. For clarity of presentation and simplicity of notation, we focus exclusively on the time-independent case. It will become evident that, for our approach, there is no fundamental technical difference between the time-dependent and time-independent scenarios. Furthermore, the notion of optimality in unraveling considered here differs from that in \cite{vicentini_optimal_2019}, which primarily studied disordered open quantum systems. Besides, this notion also differs from \cite{vovk_entanglement-optimal_2022}, which aimed to find a dynamically optimal unraveling dynamics with minimum entanglement entropy. Our approach uses different mathematical technicality and different objectives, which complements the approach in \cite{vovk_entanglement-optimal_2022}. 

Another related area is the development of quantum algorithms for simulating Lindblad equations, e.g., \cite{cleve_2017,hu_quantum_2020,schlimgen_quantum_2021,schlimgen_quantum_2022,li_simulating_2023,ding_simulating_2024,chen_randomized_2024,di2024efficient,borras2025quantum}. In particular, unraveling the Lindblad equation to quantum trajectories has been the basis for some of the quantum algorithms \cite{ding_simulating_2024,di2024efficient,borras2025quantum}.  
Until fault-tolerant quantum computers become available, the simulation of Lindblad equations still mostly relies on classical computers. Moreover, noisy quantum circuits—particularly in analog quantum machines—are typically modeled by Lindblad equations. Thus, studying the effects of quantum noise will require efficient classical solvers for Lindblad equations in the current stage. Our results are therefore relevant and potentially useful, even in the context of quantum computation.

\bigskip

{\noindent \textbf{Notations.}} We mainly consider the finite-dimensional Lindblad equation so that the Hilbert space is simply $\Complex^n$ where $n$ is the dimension of the quantum system. The unit sphere is denoted as $\mathcal{S}^n = \big\{\psi \in \Complex^n:\ \norm{\psi}^2 = 1\big\}$. The inner product is defined as $\inner{\psi}{\phi} := \sum_{i=1}^{n} \psi_i^* \phi_i$, and $A^\dagger$ means the complex transpose of $A$. By default, we use column-based vector notations. For a qubit system, $\sigma_{X}$ $\sigma_Y$, $\sigma_Z$ represent the Pauli-X,Y,Z matrices respectively; $\sigma_{+} = \begin{bsmallmatrix} 0 & 1 \\ 0 & 0 \end{bsmallmatrix}$ and $\sigma_{-} = \begin{bsmallmatrix} 0 & 0 \\ 1 & 0  \end{bsmallmatrix}$. Without specification, $\norm{\cdot}$ means the $\ell_2$ norm for wave-functions. We use the short-hand notation $\inneravg{\cdot}{\psi}$ to represent the operator's quantum average $\inner{\psi}{(\cdot) \psi}$. $\re{c}$ and $\im{c}$ means the real and imaginary parts of a complex number $c$. The expectation with respect to the classical ensemble is denoted as $\ee[\cdot]$.

\section{Optimizing unraveling schemes driven by Brownian motion}
\label{sec::sde}

We will first revisit the unraveling schemes in SDE forms in \secref{subsec::revisit}. 
Then we will provide a complete characterization of all possible such unraveling schemes in \secref{subsec::characterization} under certain assumptions, which is necessary to identify a reasonable family of SDEs to optimize. Then we will explore how the structure of drift and diffusion terms will affect the rate of change of the (classical) variance of an observable in \secref{subsec::time_deri}.

\subsection{A revisit of unraveling schemes in SDE forms}
\label{subsec::revisit}

The unraveling scheme in SDE forms was introduced in \cite{gisin_quantum-state_1992}, and one may also refer to the monograph by Percival \cite{Percival1998-sh}. We will briefly revisit this framework below. Suppose that $\psi(t)$ is a quantum state (which is possibly un-normalized), and it satisfies the following stochastic differential equation in the It\^o sense, 
    \begin{equation}\label{eqn::sde}
      \ud \psi(t)  = a\big(\psi(t)\big) \ud t + \sum_{j=1}^N b_j\big(\psi(t)\big) \ud W_j, 
    \end{equation}
on the manifold $\mathcal{M}$ (which includes all normalized wave functions, $\mathcal{M}\supseteq \mathcal{S}^n$), where $W_j$ are independent real-valued Brownian motions. The main idea of stochastic unraveling is to find appropriate $a$ and $b_j$ such that the ensemble average of $\ee\big[\psi(t)\psi(t)^\dagger\big]$ satisfies the Lindblad equation \eqref{eqn::lb}. It could be easily shown that the drift and diffusion terms must satisfy the following condition:

\begin{lemma}
\label{lem::sde_unraveling}
If $a(\psi)$, $b_j(\psi)$ satisfy
    \begin{equation}\label{eqn::unraveling_condition}
    a(\psi)\psi^\dagger + \psi a^\dagger(\psi) + \sum_{j=1}^N\ b_j(\psi)b_j^\dagger(\psi) = \mathcal{L}(\psi \psi^\dagger), \qquad \forall\ \psi \in \mathcal{M},
    \end{equation}
then $\mathbb{E}\big[\psi(t) \psi(t)^\dagger\big]$ is a solution of \eqref{eqn::lb} for all possible density matrices $\rho_{t_0}$ at arbitrary fixed initial time $t_0$. When $\mani = \mathcal{S}^n$ is the space of the unit sphere, the above sufficient condition becomes an equivalent condition.
\end{lemma}

\begin{remark}
The number of noise terms $N$ represents the dimensionality of the real-valued Brownian motion, which serves as a free parameter.
While it is convenient and a common practice to pick $N = K$, this relation is not strictly required; for instance, the cQSD \eqref{eqn::cqsd} employs $N = 2 K$. For most of the unraveling schemes considered below, we adopt the convenient convention of $N = K$ to minimize the number of free parameters.
\end{remark}

This lemma is straightforward to verify and thus the detailed proof is skipped. Hence, finding a dynamics for stochastic wave functions becomes a task of finding the drift and diffusion coefficients.
In previous literatures, a few important examples are
\begin{itemize}
    \item Linear quantum state diffusion (LQSD) \cite[Chp.~4.5]{Percival1998-sh}:
   \begin{align}
   \label{eqn::lqsd}
    \left\{\ 
   \begin{aligned}
a(\psi) &= \Big( -i {H} - \frac{1}{2} \sum_k L_k^\dagger L_k \Big) \psi, \\
b_{k}(\psi) &= L_k \psi.
\end{aligned}
\right.
\end{align}
In this case, $a(\psi)$ produces a non-Hermitian dynamics. 
    
\item Quantum state diffusion (rQSD) \cite{gisin_quantum-state_1992,Percival1998-sh}:
   \begin{align} 
\label{eqn::qsd}
  \text{(rQSD)} \qquad \left\{\ 
   \begin{aligned}
a(\psi) &= \left(-i {H} + \sum_k \langle \psi, L_k^\dagger \psi \rangle L_k 
- \frac{1}{2} L_k^\dagger L_k - \frac{1}{2} \left| \langle \psi, L_k^\dagger \psi \rangle \right|^2 \right) \psi, \\
b_{k}(\psi) &= \big(L_k - \langle \psi, L_k \psi \rangle \big)\psi.
\end{aligned}
\right.
\end{align}
This choice preserves the pathwise norm of the wave function. 

In the original literature, QSD was presented using complex-valued Brownian motions, and these can be absorbed into the function $b_k$ so that for each $k$, the term $b_k$ is replaced by two terms
\begin{align}
\label{eqn::cqsd}
 \text{(cQSD)}\qquad b_{k,1} = \frac{b_k}{\sqrt{2}}, \qquad b_{k,2} = i\frac{b_k}{\sqrt{2}}.
\end{align}

\item  Homodyne measurement \cite{albarelli_pedagogical_2024,vovk_entanglement-optimal_2022}:
   \begin{align} 
\label{eqn::qsd_homo}
  \text{(homodyne)} \qquad \left\{\ 
   \begin{aligned}
a(\psi) &= -i {H}\psi  - \frac{1}{2}\sum_{k=1}^K  \left(\begin{aligned} & L_k^\dagger L_k + \frac{1}{4} \abs{\inneravg{L_k e^{i \theta_k} + L_k^\dagger e^{-i \theta_k}}{\psi}}^2 \\
& \qquad - \inneravg{L_k e^{i \theta_k} + L_k^\dagger e^{-i \theta_k}}{\psi} e^{i \theta_k} L_k\end{aligned}\right) \psi, \\
b_{k}(\psi) &= \big(e^{i \theta_k} L_k - \frac{\inneravg{L_k e^{i \theta_k} + L_k^\dagger e^{-i \theta_k}}{\psi}}{2}\big) \psi.
\end{aligned}
\right.
\end{align}
where the parameter $\theta_k$ corresponds to the quadrature in continuous homodyne measurement.

\end{itemize}

As complex-valued Brownian motion could also be represented via real-valued Brownian motion, we only consider and present real-valued Brownian motions here for mathematical simplicity. The open question that we will investigate below is how to better choose the drift $a$ and diffusion coefficients $b_j$ so that the stochastic unraveling scheme performs better with a smaller statistical error. To address this question, we first need to answer how to characterize the family of possible unraveling schemes in the form \eqref{eqn::sde}, and then we optimize over this family.

\subsection{An almost complete characterization of stochastic unraveling in SDE forms}
\label{subsec::characterization}

Due to the possibility of introducing certain redundancy in $b_j$, we consider the following case to characterize the family of possible drift and diffusion terms that satisfy \eqref{eqn::unraveling_condition}. 

\begin{theorem}[A complete characterization for the case of one Lindblad operator]
\label{lem::characterization}
Suppose that we consider only one Lindblad operator $L$, so that the generator is 
\begin{align*}
\mathcal{L}(\rho) = - i \comm{H}{\rho} + L \rho L^\dagger - \frac{1}{2} L^\dagger L \rho - \frac{1}{2} \rho L^\dagger L,
\end{align*}
and only choose an SDE of the following form:
\begin{align}
\label{eqn::sde_1}
\ud \psi(t) = a\big(\psi(t)\big) \ud t + b\big(\psi(t)\big) \ud W(t),
\end{align}
where $a, b$ are assumed to be continuous functions. Then, 
\begin{enumerate}[(i)]
\item All possible $a$ and $b$ that satisfy \eqref{eqn::unraveling_condition} can be parameterized via the following way almost everywhere for both manifold $\mathcal{M} = \Complex^n$ and $\mathcal{S}^n$,
\begin{align}
\label{eqn::ab_Cn}
\begin{aligned}
a(\psi) &= -i {H}\psi - \frac{1}{2} L^\dagger L\psi + \big(-\frac{1}{2} \abs{{\eta(\psi)}}^2 + i {\gamma(\psi)} \big) \psi - e^{i \theta(\psi)} \eta(\psi)^* L \psi ,\\
b(\psi) &= \eta(\psi) \psi + e^{i \theta(\psi)} L \psi,
\end{aligned}
\end{align}
where $\eta: \mathcal{M} \to \Complex$ is a complex-valued function, and $\theta, \gamma: \mathcal{M} \to \Real$ are real-valued functions. \medskip

\item In particular, the above form will preserve the pathwise norm of $\psi$ when restricted to the manifold $\mathcal{S}^n$ iff
\begin{align}
\label{eqn::norm_condition}
\eta(\psi) = i h(\psi) - e^{i\theta(\psi)} \inneravg{L}{\psi},
\end{align}
where $h:\mathcal{S}^n \to \Real$ is a real-valued function.
 \medskip
\end{enumerate}
\end{theorem}

This theorem identifies a complete characterization of all possible schemes that satisfy \eqref{eqn::unraveling_condition} in a minimal setting, and in particular, a complete characterization of all possible unraveling schemes in SDE forms \eqref{eqn::sde_1} with pathwise norm-preservation. 
The proof is postponed to \appref{proof::lem::characterization}. This result is not fundamentally different from the derivation of \cite{gisin_quantum-state_1992}, but in their original paper, the prefactor $e^{i \theta}$ was neglected and overlooked, which turns out to be very useful in improving the performance of stochastic unraveling. Physically, this parameter $\theta$ connects to the quadrature in homodyne measurement \eqref{eqn::qsd_homo}; other parameters like $\gamma$ and $h$ are purely artificial parameters arising from mathematical derivation.

The above general form \eqref{eqn::ab_Cn} encompasses these typical examples in literature:
\begin{itemize}
\item When $\theta = 0$, $\eta = 0$, $\gamma = 0$, the above \eqref{eqn::ab_Cn} reduces to LQSD \eqref{eqn::lqsd}; 
\item When  $\theta = 0$, $\eta = -\inner{\psi}{L\psi}$, $\gamma = 0$, it reduces to the rQSD \eqref{eqn::qsd}.
\item When norm is preserved and $h = 0$, then 
\begin{align*}
b(\psi) = e^{i \theta(\psi)} \big(L\psi -\inneravg{L}{\psi} \psi \big).
\end{align*}
This is simply the rQSD with a phase factor $e^{i \theta(\psi)}$. However, since this $\theta$ is a state-dependent function, it is different from simply using a complex-valued Brownian motion (which is state-independent) \cite{gisin_quantum-state_1992}.

\item When one chooses
\begin{align}
\label{eqn::homodyne}
\theta(\psi) \equiv \theta, \qquad \gamma \equiv 0, \qquad \eta(\psi) = -\frac{\inneravg{L e^{i \theta} + L^\dagger e^{-i \theta}}{\psi}}{2} \equiv - \re{\inneravg{L}{\psi} e^{i \theta}},
\end{align}
then it refers to the form of homodyne measurement \cite{albarelli_pedagogical_2024}.
Mathematically, such a $\theta$ can be state-dependent, and this family of SDE has been considered in \cite{vovk_entanglement-optimal_2022} to find an dynamically optimal dynamics that minimizes the entanglement entropy with state-dependent $\theta$.

\end{itemize}

\begin{remark}
Note that a term like $i \gamma(\psi) \psi$ will only change the phase in the state vector $\psi$ (as will be more clear in the proof in Appendix~\ref{proof::equiv_opti}, and will not affect the observable. If we only worry about the observable, then the effective dynamics (by tracing out these redundant parameters) that \emph{also} has pathwise norm preservation will be
\begin{align*}
a(\psi) &= -i {H}\psi - \frac{1}{2} L^\dagger L\psi  -\frac{1}{2} \abs{{\eta(\psi)}}^2 \psi - e^{i \theta(\psi)} \eta(\psi)^* L \psi ,\\
b(\psi) &= \eta(\psi) \psi + e^{i \theta(\psi)} L \psi, \qquad  \eta(\psi) = - \re{\inneravg{L}{\psi} e^{i \theta(\psi)}}.
\end{align*}
This is essentially the form for homodyne measurement. This establishes an interesting fact that the effective family of stochastic unraveling scheme with pathwise norm preservation in the above SDE form is essentially the one for homodyne measurement with state-dependent quadrature.
\end{remark}

Due to the linearity of the Lindblad equation, it is natural to consider the following family of stochastic unraveling schemes for $N$ Lindblad operators:
\begin{corollary}
The following family of drift and diffusion terms are admissible stochastic unraveling schemes:
\begin{align}
\label{eqn::ab_Cn::2}
\begin{aligned}
a(\psi) &= -i {H}\psi + \sum_{k=1}^{{K}}\ \Big(-\frac{1}{2} L_k^\dagger L_k \psi + \big(-\frac{1}{2} \abs{\eta_k(\psi)}^2 + i {\gamma}_k(\psi) \big) \psi - e^{i \theta_k(\psi)} \eta_k(\psi)^* L_k \psi\Big) , \\
b_k(\psi) &= \eta_k(\psi)\ \psi + e^{i \theta_k(\psi)} L_k \psi,
\end{aligned}
\end{align}
where for each $k = 1, 2, \cdots, K$, $\eta_k$ is a complex-valued continuous function, and $\theta_k$, $\gamma_k$ are real-valued continuous functions. To preserve the norm, we need
\begin{align}
\label{eqn::norm}
\eta_k(\psi) = i h_k(\psi) - e^{i\theta_k(\psi)} \inneravg{L_k}{\psi},
\end{align}
where each $h_k$ is a real-valued continuous function.
\end{corollary}

This conclusion immediately follows from the linearity of Lindblad generator \eqref{eqn::lb} and Theorem~\ref{lem::characterization}, and thus a detailed proof is omitted.

\subsection{Dynamically optimal quantum state diffusion}
\label{subsec::time_deri}

For any unraveling scheme satisfying \eqref{eqn::unraveling_condition}, one can simulate the SDE \eqref{eqn::ab_Cn::2} to recover $\rho(t)$ and approximate
\begin{align}
    \label{eqn::observable}
\tr\big(\obs \rho(t)\big) \approx \frac{1}{\numsample} \sum_{j=1}^{\numsample} \avgbraket{\obs}{\psi^{(j)}(t)},
\end{align}
where each $\psi^{(j)}(t)$ is an \emph{i.i.d.} random variable obtained by solving \eqref{eqn::ab_Cn::2}, $\obs = \obs^\dagger$ is the observable of interest, and $\numsample$ is the number of samples.
If we temporarily ignore numerical discretization errors (which can be minimized using advanced integrators), all such unraveling schemes yield an unbiased estimate of $\tr\big(\obs \rho(t)\big)$. The main difference between these schemes lies in the variance of the stochastic estimates. Although the original QSD scheme \cite{gisin_quantum-state_1992} already performs well in many examples studied in the literature, our goal is to investigate which scheme is optimal in this regard.

We first note that the variance of the stochastic unraveling scheme is
\begin{align*}
\text{Var}(a, b, t) &= \ee \Big[\absbig{\avgbraket{\obs}{\psi(t)}}^2 \Big] - \absbig{\ee\big[ \avgbraket{\obs}{\psi(t)}\big]}^2 \\
& \equiv \ee \Big[\absbig{\avgbraket{\obs}{\psi(t)}}^2\Big] - \abs{\tr(\obs \rho(t))}^2.
\end{align*}
This variance merely comes from the classical fluctuation arising from the numerical simulation, rather than the quantum statistical variance \cite{BREUER_1996}.
For a given Lindblad equation to simulate, the mean of the observable is always fixed, and thus it is only necessary to optimize the second-moment above.

We formulate the problem as
\begin{align}
\label{eqn::deri_var}
\begin{aligned}
\mathsf{dV}_{\text{diffusion}} = &\ \text{minimize}\ \frac{\ud}{\ud t}\ \ee \absbig{\avgbraket{\obs}{\psi(t)}}^2 \\
&\text{subject to functions}\ a, b \text{ following the diffusion ansatz in } \eqref{eqn::ab_Cn::2}.
\end{aligned}
\end{align}
This minimizes the increment of variance locally in time $t$. This is a greedy-type dynamical optimization approach without considering the global effect, in a way similar to the dynamical low-rank approximation \cite{koch_dynamical_2007,le_bris_low-rank_2013,cao_stochastic_2018}.

\begin{theorem}[Dynamically optimal quantum state diffusion]
\label{thm::optimal_soln}
The optimization problem \eqref{eqn::deri_var} for a general SDE ansatz in \eqref{eqn::ab_Cn::2} admits the following explicit global minimizers $(\eta_k^\star, \theta_k^\star)$, where their expressions are given by
\begin{align}
\label{eqn::do_qsd}
{\normalfont \text{(DO-QSD)}}\qquad\qquad 
\left\{\ 
\begin{aligned}
\eta_k^\star(\psi) =&\ - e^{i \theta_k^\star(\psi)} \inneravg{L_k}{\psi}, \\
e^{i \theta^\star_k(\psi)} =&\ i e^{-i \phase_k(\psi)},
\end{aligned}\right.
\end{align}
where $\phase_k(\psi)$ is the phase of $\inneravg{\obs L_k}{\psi} - \inneravg{L_k}{\psi} \inneravg{\obs}{\psi}$. 
The optimal value is 
\begin{align}
\label{eqn::p_optimal}
\mathsf{dV}_{\text{diffusion}} = 2 \ee \Big[\innerbig{\obs \psi}{\mathcal{L}(\psi\psi^\dagger)\obs \psi}\Big]  + \sum_{k=1}^K\ 2 \ee\Big[\inner{L_k\psi}{\obs L_k\psi}\inner{\psi}{\obs\psi} - \abs{\inner{\obs\psi}{L_k\psi}}^2\Big] \qquad \psi \equiv \psi(t).
\end{align}
\end{theorem}

The proof is postponed to \appref{proof::equiv_opti}. We make a few remarks below:
\begin{itemize}

\item (\emph{Computational costs}). The term $\inneravg{L_k}{\psi}\equiv \inner{\psi}{L_k\psi}$ also needs to be computed in the conventional QSD, so that there is no extra cost in computing $L_k {\psi}$. Another term is $\obs {\psi}$, which also needs to be computed each time when one uses the stochastic unraveling, as it is needed in computing the observable on-the-fly \eqref{eqn::observable}. Therefore, the additional computational cost is negligible, while the reduction in variance sometimes could be significant, as will be demonstrated in \secref{sec::numerics}.

\item (\emph{Multiple observables}). The generalization to multiple observables $\obs$ can be formulated similarly; see \appref{appx::multiple_observable}. However, since the phase factor must accommodate different observables, achieving substantial improvements over conventional QSD \eqref{eqn::qsd} may become more challenging. Our primary interest in this paper is the case where a single observable is considered.

\item (\emph{Other properties}). The rQSD \eqref{eqn::qsd} and the cQSD (with complex noise) \eqref{eqn::cqsd} are both sub-optimal in this sense, as will be numerically demonstrated in \secref{sec::numerics}. Furthermore, we observe that the role of $h$ is negligible in terms of variance reduction. The optimal solution $\theta_k^\star$ is invariant under the addition of a constant to $\obs$; that is, it remains unchanged for the entire family $\{\obs + c \unit\ :\ c\in \Real\}$, which is certainly also physically reasonable.

\end{itemize}

\section{Optimizing unraveling schemes driven by Poisson process}
\label{sec::sjp}

The stochastic unraveling could be achieved not only by an SDE process, but also by the jump process \cite{dalibard_wave-function_1992}. We will revisit the unraveling schemes in jump process forms in \secref{subsec::revisit_jump}. Then we will similarly characterize all possible unraveling schemes in jump processes when there is only one Lindblad operator with one jump operator in \secref{subsec::characterization_qjp}. We will subsequently optimize the local rate of change of the variance so that the optimal jump scheme is achieved in \secref{sec::do_qjp}.

\subsection{A revisit of unraveling schemes in jump process}
\label{subsec::revisit_jump}

Similar to the above SDE cases, we could also consider piecewise-deterministic Markov processes (PDMP) \cite{breuer_theory_2007} in the following form:
\begin{align}
\label{eqn::sjp}
\ud \psi(t) = a\big(\psi(t)\big) \ud t + \sum_{j=1}^{N} b_j\big(\psi(t)\big) \ud \NP_j(t),
\end{align}
where $\NP_j(t)$ are independent Poisson processes with state-dependent rate $\jumprate_j\big(\psi(t)\big)$. 
Numerically, to simulate the above dynamics, we can use 
\begin{align}
    \label{eqn::jump_numerical}
    \begin{aligned}
\psi(t+\dt) \approx \left\{
\begin{aligned}
&\ \psi(t) + a(\psi(t))\ \dt \qquad & \text{ with probability } 1 - \sum_{j} \jumprate_j(\psi(t)) \dt;\\
&\ \psi(t) + b_j(\psi(t)), \qquad & \text{ with probability } \jumprate_j(\psi(t)) \dt \text{ for each } j.
\end{aligned}\right.
\end{aligned}
\end{align}
The classical Monte-Carlo Wave Function (MCWF) approach \cite{dalibard_wave-function_1992}, which can also be called a quantum jump process, refers to the following choices
with the number of jump operators chosen as $N = K$: 
\begin{align}
\label{eqn::WFMC}
\text{(QJP)}\qquad \left\{
\begin{aligned}
a(\psi) &= - i (H - \frac{i}{2} \sum_{k=1}^K L_k^\dagger L_k)\psi + \frac{p(\psi)}{2} \psi, \qquad p(\psi) = \sum_{k=1}^K \inneravg{L_k^\dagger L_k}{\psi},\\
b_k(\psi) &= \frac{L_k \psi}{\norm{L_k \psi}} - \psi, \qquad \jumprate_k(\psi) = \inneravg{L_k^\dagger L_k}{\psi}, \qquad k = 1, 2 \cdots, K.
\end{aligned}\right.
\end{align}
The number $p(\psi)$ is chosen to ensure norm preservation in the absence of jumps; in the literature, this step is typically presented algorithmically, with normalization applied as a post-processing step.
It can be verified that $\rho(t) := \ee\big[\psi(t) \psi(t)^\dagger\big]$ satisfies the Lindblad equation \eqref{eqn::lb}; see also \cite{dalibard_wave-function_1992}. Similar to the quantum state diffusion scheme in \eqref{eqn::qsd}, the jump process scheme \eqref{eqn::WFMC} has also been widely used in the literature. However, to the best of the authors' knowledge, there has been little discussion of all possible such realizations.
We will provide an equivalent condition for the jump process to be a stochastic unraveling of the Lindblad equation with norm preservation.

\begin{lemma}
\label{lem::sjp_unraveling}
    Assume that $\psi(t)$ satisfies the above stochastic jump process \eqref{eqn::sjp} with the regularity assumptions: $a$, $\jumprate_j$, $\sqrt{\jumprate_j} b_j$ are continuous functions on $\mathcal{S}^n$, and $b_j$ is a bounded function on $\mathcal{S}^n$ for any index $j$. Then it is a stochastic unraveling of the Lindblad equation with norm preservation iff 
    \begin{align}
    \label{eqn::qjp_condition}
    \begin{aligned}
    & \text{(norm-preservation):} \\
    &\qquad \inner{\psi}{a(\psi)} + c.c. = 0,\qquad  \forall \psi\in \mathcal{S}^n, \\
    & \qquad \norm{b_j(\psi) + \psi} = 1, \qquad \forall \psi\in \mathcal{S}^n \text{such that } \jumprate(\psi) >0,\ \text{ and }\forall 1\le j\le N, \\
    & \text{(unraveling constraint):} \\
    &\qquad a(\psi) \psi^\dagger + \psi a(\psi)^\dagger + \sum_{j=1}^N \jumprate_j(\psi) \Big(b_j(\psi) \psi^\dagger + \psi b_j(\psi)^\dagger + b_j(\psi) b_j^\dagger(\psi)\Big) = \mathcal{L}(\psi \psi^\dagger).
    \end{aligned}
    \end{align}
\end{lemma}
When $\jumprate_j(\psi) = 0$, there is no jump at all, so that there won't be any restriction on $b_j$. The proof is given in \appref{proof::lem::sjp}.
It could be directly verified that the above QJP \eqref{eqn::WFMC} satisfies all these constraints.

\subsection{An almost complete characterization of stochastic unraveling in jump process}
\label{subsec::characterization_qjp}

We observe that the above constraint exhibits a structure similar to the It\^o SDE case in \eqref{eqn::unraveling_condition}. 
Therefore, we can likewise provide a complete characterization of all possible jump processes in the same manner.

\begin{theorem}[A complete characterization for the case of one Lindblad operator] 
\label{thm::qjp_characterize}
Suppose that we consider only one Lindblad operator $L$, so that the generator is 
\begin{align*}
\mathcal{L}(\rho) = - i \comm{H}{\rho} + L \rho L^\dagger - \frac{1}{2} L^\dagger L \rho - \frac{1}{2} \rho L^\dagger L,
\end{align*}
and suppose that we only choose the (minimal) jump process in the following form:
\begin{align}
\label{eqn::sjp_1}
\ud \psi(t) = a\big(\psi(t)\big) \ud t + b\big(\psi(t)\big) \ud \NP(t),
\end{align}
where $\NP(t)$ is a state-dependent Poisson process with rate $\jumprate(\psi(t)\big)$. Then all possible $a$, $b$, $\jumprate$ that satisfy the constraints in \eqref{eqn::qjp_condition} can be described as follows:
\begin{align}
\label{eqn::opti_jump}
\left\{
\begin{aligned}
a(\psi) &=  -i {H}\psi - \frac{1}{2} L^\dagger L\psi + \big(-\frac{1}{2} \abs{{\eta(\psi)}}^2 + i {\gamma(\psi)} \big) \psi \\
&\qquad \qquad - e^{i \theta(\psi)} \eta(\psi)^* L \psi - \sqrt{\jumprate(\psi)} \big(\eta(\psi)\ \psi + e^{i \theta(\psi)} L \psi\big) , \\
b(\psi) &= \frac{\eta(\psi)}{\sqrt{\jumprate(\psi)}} \psi + \frac{e^{i\theta(\psi)}}{\sqrt{\jumprate(\psi)}} L \psi , \\
\eta(\psi) &= e^{i \beta(\psi)} \sqrt{\jumprate(\psi) + \abs{\inneravg{L}{\psi}}^2 - \inneravg{L^\dagger L}{\psi}} - e^{i\theta(\psi)} \inneravg{L}{\psi} - \sqrt{\jumprate(\psi)}, \\
\jumprate(\psi) &=  \inneravg{L^\dagger L}{\psi} - \abs{\inneravg{L}{\psi}}^2 + \alpha(\psi),\\
\end{aligned}
\right.
\end{align}
holds almost everywhere on $\mathcal{S}^n$, where $\alpha \ge 0$ is an arbitrary non-negative function, $\beta, \theta, \gamma$ are arbitrary real-valued functions.
\end{theorem}

The quantity $\lambda$ is the jump rate, and $\alpha$ measures the increment of this rate relative to the quantum variance of the operator $L$ (though $L$ may not be Hermitian herein). The parameter $\theta$ plays the similar role as the measured quadrature in homodyne detection, but it is introduced here primarily for mathematical convenience. The quantities $\gamma$ and $\beta$ are degrees of freedom that do not play any role in estimating the observables (as will be clear in later), so they can be treated as purely artificial parameters arising from mathematical treatment.

We can easily verify that the above QJP \eqref{eqn::WFMC} with $K = 1$ refers to the following special choice:
\begin{align}
\left\{
\begin{aligned}
& \jumprate(\psi) = \inneravg{L^\dagger L}{\psi} \qquad \Longleftrightarrow\qquad  \alpha(\psi) =  \abs{\inneravg{L}{\psi}}^2, \\
& \eta(\psi) = -\sqrt{\jumprate(\psi)}, \qquad \Longleftrightarrow\qquad e^{i \beta(\psi)} = \frac{\inneravg{L}{\psi}}{\abs{\inneravg{L}{\psi}}}, \\
& \theta(\psi) = \gamma(\psi) = 0. \\
\end{aligned}\right.
\end{align}
The proof is postponed to \appref{proof::qjp_characterize}. When there are more than one Lindblad operator and possibly with more $b_j$, it appears   challenging to provide a complete characterization due to the possibility of introducing redundant terms. Though not being the family of all possible such jump processes, the following family of quantum jump processes still provides a very large family of possible stochastic unraveling in the form of jump processes, which becomes a complete characterization in the sense of Theorem~\ref{thm::qjp_characterize}.

\begin{corollary}
The following family of drift and jump terms provide valid stochastic unraveling schemes:
\begin{align}
\label{eqn::jump_general_form}
\left\{
\begin{aligned}
a(\psi) &=  -i {H}\psi - \frac{1}{2} \sum_{k=1}^K L_k^\dagger L_k \psi +  \sum_{k=1}^{K}\big(-\frac{1}{2} \abs{{\eta_k(\psi)}}^2 + i {\gamma}_k(\psi) \big) \psi \\
&\qquad\qquad - \sum_{k=1}^K e^{i \theta_k(\psi)} \eta_k(\psi)^* L_k \psi  - \sum_{k=1}^K \sqrt{\jumprate_k(\psi)} \big(\eta_k(\psi) \psi + e^{i \theta_k(\psi)} L_k \psi\big), \\
b_k(\psi) &= \frac{\eta_k(\psi)}{\sqrt{\jumprate_k(\psi)}} \psi + \frac{e^{i\theta_k(\psi)}}{\sqrt{\jumprate_k(\psi)}} L_k \psi, \qquad k = 1, 2, \cdots, K \\
\eta_k(\psi) &= e^{i \beta_k(\psi)} \sqrt{\jumprate_k(\psi) + \abs{\inneravg{L_k}{\psi}}^2 - \inneravg{L_k^\dagger L_k}{\psi}} - e^{i\theta_k(\psi)} \inneravg{L_k}{\psi} - \sqrt{\jumprate_k(\psi)}, \\
\jumprate_k(\psi) &=  \inneravg{L_k^\dagger L_k}{\psi} - \abs{\inneravg{L_k}{\psi}}^2 + \alpha_k(\psi),\\
\end{aligned}
\right.
\end{align}
where $\alpha_k\ge 0, \beta_k, \theta_k, \gamma_k$ are arbitrary real-valued continuous functions.
\end{corollary}

The proof is straightforward since we only need to validate that this family of drift and jump terms satisfy \eqref{eqn::qjp_condition}.

\subsection{Dynamically optimal quantum jump process}
\label{sec::do_qjp}

Similar to the SDE case, we seek the optimal parameters $\alpha$, $\beta$, $\theta$, and $\gamma$ that minimize the rate of change of the variance locally in time. Equivalently, we solve the following optimization problem:
\begin{align}
\label{eqn::deri_var_jump}
\begin{aligned}
\mathsf{dV}_{\text{jump}} :=\ & \text{minimize}\ \frac{\ud}{\ud t}\ \ee\  \abs{\inner{\psi(t)}{\obs\psi(t)}}^2\\
&\text{subject to functions } a, b \text{ following the jump process ansatz in } \eqref{eqn::opti_jump}\\
& \qquad \jumprate \text{ is uniformly bounded by } \Lambda,
\end{aligned}
\end{align}
where the positive number $\Lambda > 0$ is a parameter to choose, and we assume that $\Lambda$ is chosen reasonably large enough so that the above $\alpha_k\ge 0$ exists. The reason for imposing such an upper bound constraint is that for the jump process, one requires $\jumprate \Delta t \ll 1$. Since we cannot in practice choose an infinitesimally small $\Delta t$, one would necessarily need a certain  uniform upper bound for the jump rate in practice.

\begin{theorem}
\label{thm::optimal_soln_jump}
For the quantum jump process, 
\begin{itemize}
\item[(i)] For the optimization problem of minimizing the local rate of change of variance, the quantum state diffusion ansatz is always at least as effective as the jump process locally in time:
\begin{align}
\label{eqn::comparison}
\mathsf{dV}_{\text{diffusion}} \le \mathsf{dV}_{\text{jump}}.
\end{align}

\item[(ii)] Without loss of generality, consider $K = 1$. 
The optimization problem \eqref{eqn::deri_var_jump} for a general jump process ansatz in \eqref{eqn::jump_general_form} admits explicit global minimizers $(\alpha^\star, \theta^\star)$, 
where the expression for $\theta^\star$ is given by \eqref{eqn::theta_optimal-proof}, \eqref{eqn::ABC}, and the expression for $\jumprate^\star = \Lambda$ (which immediately determines $\alpha^\star$).
Furthermore, we may set $\gamma = 0$, as it does not affect the variance, and $\beta = 0$, since the loss function in \eqref{eqn::deri_var_jump} depends only on the phase difference $\theta - \beta$.
\end{itemize}
\end{theorem}

The detailed proof is deferred to \appref{proof::optimal_soln_jump}. The parameter $\gamma$ can be set to zero, as this phase factor cancels out when reconstructing the density matrix $\ee\big[\psi(t){\psi(t)}^\dagger \big]$. The notable message of the above theorem is that, at least locally, the diffusion ansatz can generally be tuned to perform at least as well as the jump process in terms of local empirical variance growth. Moreover, optimizing the local variance growth for the jump process appears to be substantially more complex than that for the diffusion process, which further motivates favoring the diffusion ansatz for stochastic unraveling. However, this comparison does not account for numerical discretization errors or the need for physical interpretation. The discussion here focuses solely on the local performance of unraveling schemes, assuming sufficiently accurate numerical solvers for the stochastic process. 
In general, there exists some Lindblad equation and certain parameter region of $\Lambda$, so that the above inequality \eqref{eqn::comparison} is strict.

\section{Numerical examples}
\label{sec::numerics}

In this section, we consider Lindblad equations \eqref{eqn::lb} for (a) a 2D decaying model; (b) a cavity QED model for demonstration. We will compare conventional rQSD \eqref{eqn::qsd}, cQSD \eqref{eqn::cqsd}, homodyne \eqref{eqn::homodyne} (with $\theta=0$) and our new scheme DO-QSD \eqref{eqn::do_qsd}, together with a machine learning trained one for the first example using \eqref{eqn::ab_Cn::2} and \eqref{eqn::norm}. The results below demonstrate that the new training-free DO-QSD has overall superior performance compared with conventional choices, and it has comparable and even superior performance compared with the machine learning trained one.

In the numerical experiments, for each unraveling scheme, we use $\numsample$ samples to obtain an estimator $X_i \approx \tr\big(\obs\rho(t_i))$, where $\{t_i\}_{i=0}^{\numtime}$ are time grid points. Then the trajectory-based error and the trajectory-based variance are defined as follows:
\begin{align}
\label{eqn::error_var}
\left\{\begin{aligned}
\text{Trajectory error} &= \frac{1}{\numtime+1} \sum_{i=0}^{\numtime}  \abs{X_i - \tr\big(\obs\rho(t_i))}, \\
\text{Averaged var} &=  \frac{1}{\numtime+1}  \sum_{i=0}^{\numtime} \text{Var}(X_i).
\end{aligned}\right.
\end{align}
The SDE \eqref{eqn::ab_Cn::2} is simulated using Platen's weak second-order algorithm \cite[Chp.~5.2.A]{kloeden}; since the norm is theoretically always preserved for the above mentioned choices of unraveling scheme, we enforce the normalization as a post-processing step to maintain numerical stability. For these different unraveling schemes, we fix the same time step and the same amount of sample size, and then we estimate the above two quantities to compare their performance in terms of the actual error (influenced by both bias and the stochastic fluctuation) and the variance (mainly from the stochastic fluctuation). To ensure robustness, we repeat these experiments $10$ times and will report the boxplot of trajectory error and averaged variance in \eqref{eqn::error_var}.

We would like to acknowledge that the extent that the variance reduction could be achieved is certainly problem dependent, which relies on all three factors: initial condition, Lindblad equation, as well as the observable; when the increment of variance due to the dynamics itself (cf. \eqref{eqn::p_optimal}) is dominating compared to the reduction due to choice of unraveling, then there is indeed little room for variance reduction. For these two benchmark examples, we observe a promising performance for DO-QSD for certain parameter regions and observables. 
Since the expression of DO-QJP is complex and also due to the theoretical guarantee in Theorem~\ref{thm::optimal_soln_jump}, we mostly test the diffusion-based ansatz herein. However, we will compare our approach (DO-QSD) with the standard implementation of QuTiP \cite{lambert_qutip_2026} using the function called mcsolve, which implements the conventional quantum jump process \eqref{eqn::WFMC}. 
As a remark, due to the Monte Carlo nature of sampling method, to achieve the same precision, the amount of saving could be roughly captured by the amount of reduction in variance, provided that we use a high-order efficient numerical solver.

\subsection{A 2D decaying model}
\label{subsec::eg1}

We begin by examining a simple two-level Lindblad system, adapted from Eq.~(3.219) of \cite{breuer_theory_2007}:
\begin{align}
\label{eg::two_decay}
H = 0,\qquad  
L_1 = \sqrt{\lambda_0(\nu+1)}\,\sigma_{-},\qquad  
L_2 = \sqrt{\lambda_0 \nu}\,\sigma_{+}.
\end{align}
For the experiments, we set $\lambda_0 = 5$, $\nu = 0.5$, use the time interval $[0,2]$, and choose the initial state $\rho_0 = \ket{+}\bra{+}$.  
Figure~\ref{fig::eg1} compares rQSD~\eqref{eqn::qsd}, cQSD~\eqref{eqn::cqsd}, the trained unraveling scheme (obtained by training $h_k$ and $\theta_k$ in \eqref{eqn::ab_Cn::2}), and DO-QSD~\eqref{eqn::do_qsd}. The detailed training procedure and hyper-parameters are deferred to \appref{sec::ml}.

As shown in Figure~\ref{fig::eg1::var}, DO-QSD consistently achieves the lowest trajectory error among rQSD and cQSD across different time steps.
Because the value for observable $\sigma_Y$ yields a trivial zero value, it is omitted from the comparison.
Figure~\ref{fig::eg1::var} also illustrates the empirical variance.  
Here, one can observe that DO-QSD largely reduces the variance, and performs comparably to the trained scheme, whose loss function is explicitly designed to minimize the empirical variance.  
Although this training-free method is not guaranteed to be globally optimal with respect to the averaged variance in \eqref{eqn::error_var}, its competitive performance is encouraging.

\begin{figure}[h!]
\centering
\begin{subfigure}[b]{\figwidth\textwidth}
\includegraphics[width=\textwidth]{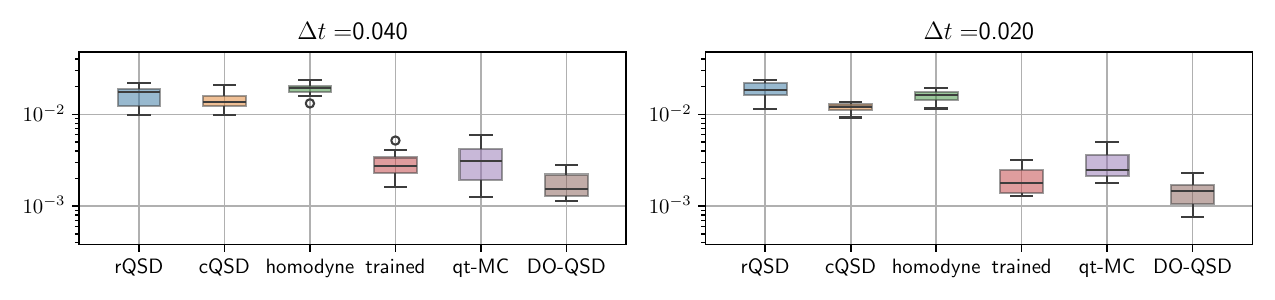}
\caption{Observable $\sigma_X$}
\end{subfigure}
~
\begin{subfigure}[b]{\figwidth\textwidth}
\includegraphics[width=\textwidth]{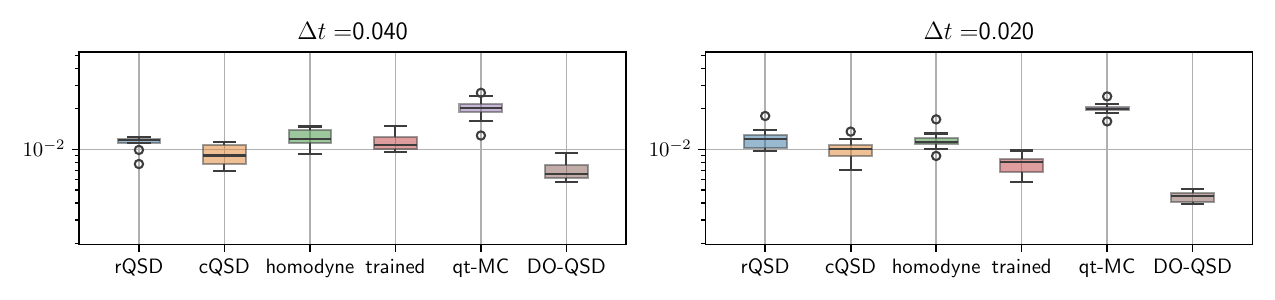}
\caption{Observable $\sigma_Z$}
\end{subfigure}
\caption{\textbf{Trajectory errors} for various unraveling schemes applied to the 2D decaying model~\eqref{eg::two_decay}.  
The rQSD~\eqref{eqn::qsd}, cQSD~\eqref{eqn::cqsd}, homodyne~\eqref{eqn::homodyne} (with $\theta=0$) methods serve as standard baselines.
The \emph{trained} method denotes the scheme obtained by training $h_k$ and $\theta_k$ in~\eqref{eqn::ab_Cn::2} with the neural-network parameterization.
The \emph{qt-MC} method uses the default solver \enquote{qutip.mcsolve} \cite{lambert_qutip_2026}, which implements quantum jump process \eqref{eqn::WFMC}.
DO-QSD (ours) corresponds to~\eqref{eqn::do_qsd}.  
Results are based on a sample size of $\numsample = 10^3$, and the boxplot summarizes $10$ independent runs.}
\label{fig::eg1}
\end{figure}

\begin{figure}[h!]
\centering
\begin{subfigure}[b]{\figwidth\textwidth}
\includegraphics[width=\textwidth]{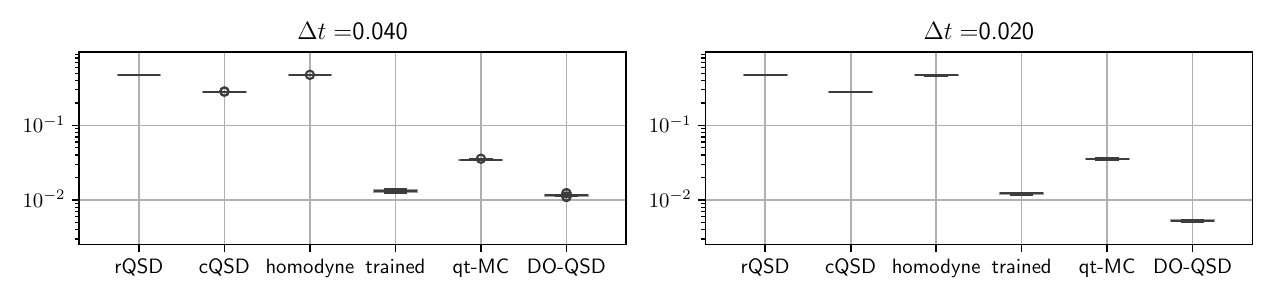}
\caption{Observable $\sigma_X$}
\end{subfigure}
~
\begin{subfigure}[b]{\figwidth\textwidth}
\includegraphics[width=\textwidth]{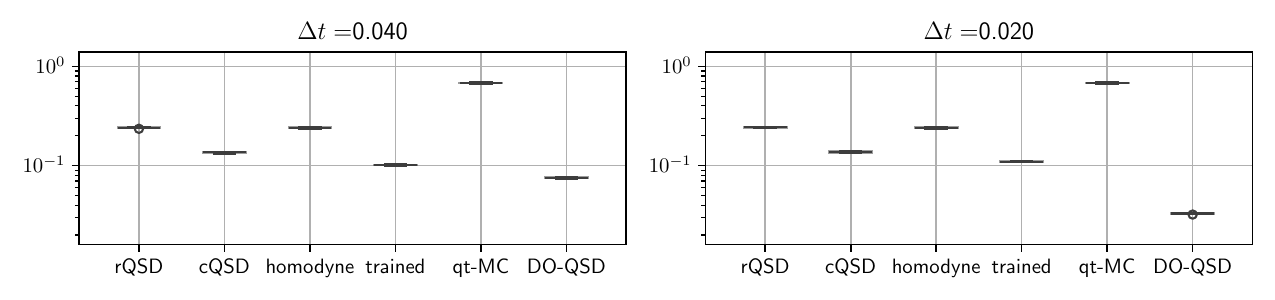}
\caption{Observable $\sigma_Z$}
\end{subfigure}
\caption{\textbf{Averaged variance} for various unraveling schemes applied to the 2D decaying model~\eqref{eg::two_decay} for the same experiments as reported in Figure~\ref{fig::eg1}.}
\label{fig::eg1::var}
\end{figure}

\subsection{A cavity QED model}

In the next example,
we consider the following composite system of a two-level atom and a photon field, adapted from \cite{briegel_quantum_1993}:
\begin{align}
\label{eg::atom::H}
H = \unit_{\text{atom}} \otimes \hat{a}^\dagger \hat{a} + \sigma_Z \otimes \unit_{\text{ph}} - (\sigma_{-} \otimes \hat{a}^\dagger + \sigma_{+} \otimes \hat{a}\big) . 
\end{align}
Here, $\hat{a}^{(\dagger)}$ are the creation and annihilation operators for the photon field, the Pauli matrices act on the two-level atom, and $\unit_{\text{atom}}$ and $\unit_{\text{ph}}$ are the identity operators for the atomic and photonic subsystems, respectively. For simplicity, the frequency and coupling strength parameters in the Hamiltonian are set to unity.

The Lindbladian part consists of five terms:
\begin{align}
\label{eqn::atom::L1}
L_1 = \unit_{\text{atom}} \otimes \sqrt{\mu_1 (\nu+1)} \hat{a}, \qquad L_2 = \unit_{\text{atom}} \otimes \sqrt{\mu_1 \nu} \hat{a}^\dagger ,
\end{align}
which describes the interaction of the photon in the cavity with a thermal reservoir, and
\begin{align}
\label{eqn::atom::L4}
L_3 = \sqrt{\mu_2 (1-\kappa)} \sigma_{-} \otimes \unit_{\text{ph}}, \qquad L_4 = \sqrt{\mu_2 \kappa} \sigma_{+} \otimes \unit_{\text{ph}}, \qquad L_5 = \sqrt{\mu_3} \sigma_Z \otimes \unit_{\text{ph}} ,
\end{align}
which models the dynamics of the atomic subsystem. For the simulations, we choose the parameters $\nu = \kappa = 1/2$ and $\mu_1 = \mu_2 = \mu_3 = 0.6$. The initial state is taken as the product state $\ket{+}\bra{+}\otimes \ket{0}\bra{0}$, meaning the atom starts in the maximally entangled state and the photon field is initially in the vacuum state. The photon field is truncated to dimension $10$, resulting in a total system dimension of $20$, and the time interval is $[0,2]$.

Simulation errors for various observables are shown in Figure~\ref{fig::eg2}, and the corresponding variances in Figure~\ref{fig::eg2::var}. For the observable $\sigma_Z\otimes \unit_{\text{ph}}$, the error is significantly reduced using the same number of samples in the stochastic unraveling. For $\unit_{\text{atom}} \otimes \hat{a}^\dagger \hat{a}$ and $\sigma_X\otimes \unit_{\text{ph}}$, the reduction in error is less significant. For  $\sigma_Y\otimes \unit_{\text{ph}}$, there is very little reduction in error. Overall, the extent of reduction in error is consistent with the extent of reduction in variance, as shown in Figure~\ref{fig::eg2::var}.
These results confirm that the extent of error reduction is problem-dependent, as discussed previously. We emphasize that, for different choices, DO-QSD consistently provides performance comparable to conventional unraveling schemes such as rQSD and cQSD. In some cases, a more significant improvement is possible, depending on the specific problem.

\begin{figure}
\centering
\begin{subfigure}[b]{0.9\textwidth}
\includegraphics[width=\textwidth]{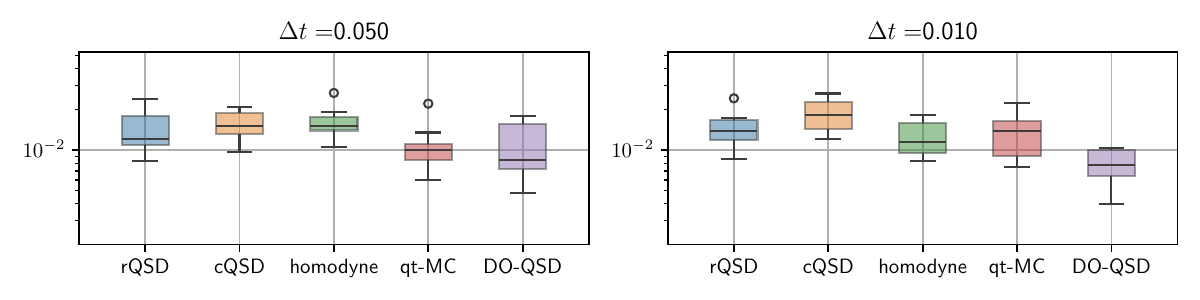}
\caption{Observable $\sigma_X\otimes \unit_{\text{ph}}$}
\end{subfigure}
\\
\begin{subfigure}[b]{0.9\textwidth}
\includegraphics[width=\textwidth]{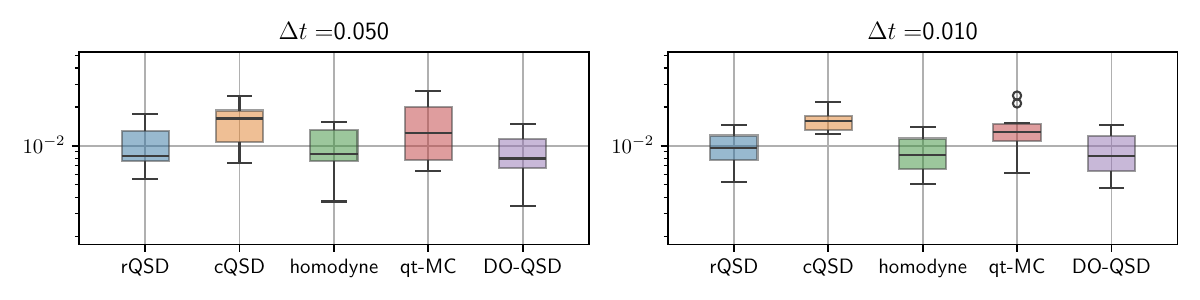}
\caption{Observable $\sigma_Y\otimes \unit_{\text{ph}}$}
\end{subfigure}
\\
\begin{subfigure}[b]{0.9\textwidth}
\includegraphics[width=\textwidth]{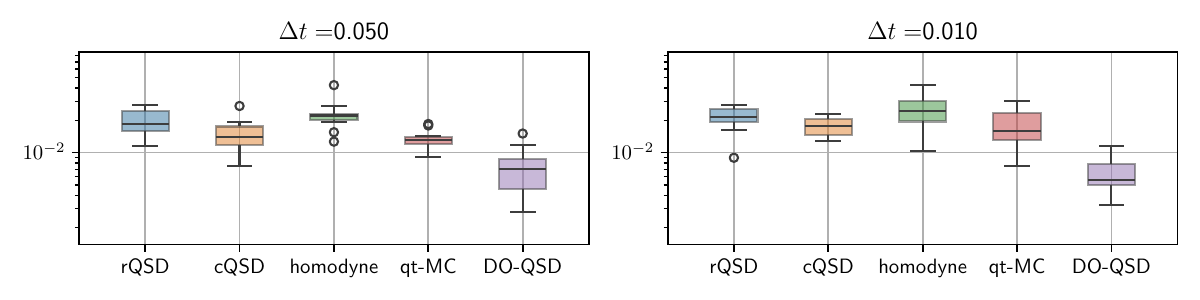}
\caption{Observable $\sigma_Z\otimes \unit_{\text{ph}}$}
\end{subfigure}
\\
\begin{subfigure}[b]{0.9\textwidth}
\includegraphics[width=\textwidth]{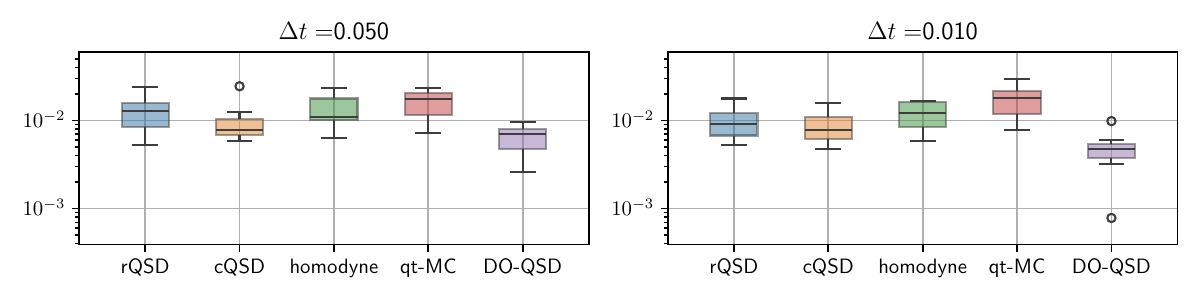}
\caption{Observable $\unit_{\text{atom}} \otimes \hat{a}^\dagger \hat{a}$}
\end{subfigure}
\caption{\textbf{Trajectory errors} for various unraveling schemes applied to the atom-cavity model~\eqref{eg::atom::H}, \eqref{eqn::atom::L1}, \eqref{eqn::atom::L4}.
The rQSD~\eqref{eqn::qsd}, cQSD~\eqref{eqn::cqsd}, homodyne~\eqref{eqn::homodyne} (with $\theta=0$) methods serve as standard baselines.
The \emph{qt-MC} method uses the default solver \enquote{qutip.mcsolve} \cite{lambert_qutip_2026}, which implements quantum jump process \eqref{eqn::WFMC}.
DO-QSD (ours) corresponds to~\eqref{eqn::do_qsd}.  
Results are based on a sample size of $\numsample = 500$, and the boxplot summarizes $10$ independent runs.}
\label{fig::eg2}
\end{figure}

\begin{figure}
\centering
\begin{subfigure}[b]{0.9\textwidth}
\includegraphics[width=\textwidth]{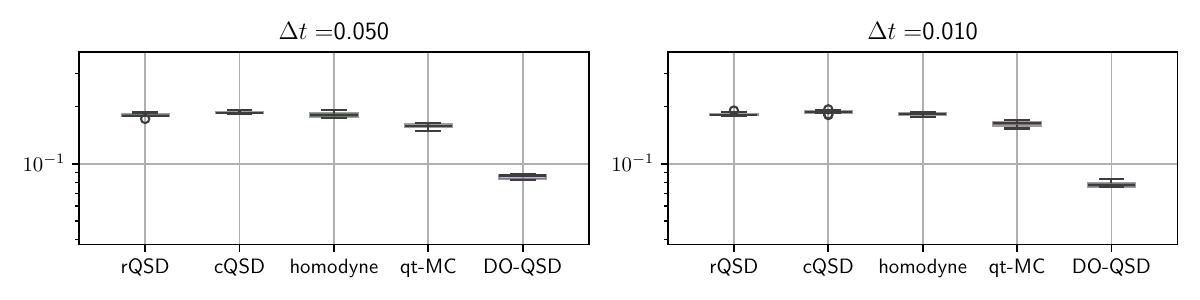}
\caption{Observable $\sigma_X\otimes \unit_{\text{ph}}$}
\end{subfigure}
\\
\begin{subfigure}[b]{0.9\textwidth}
\includegraphics[width=\textwidth]{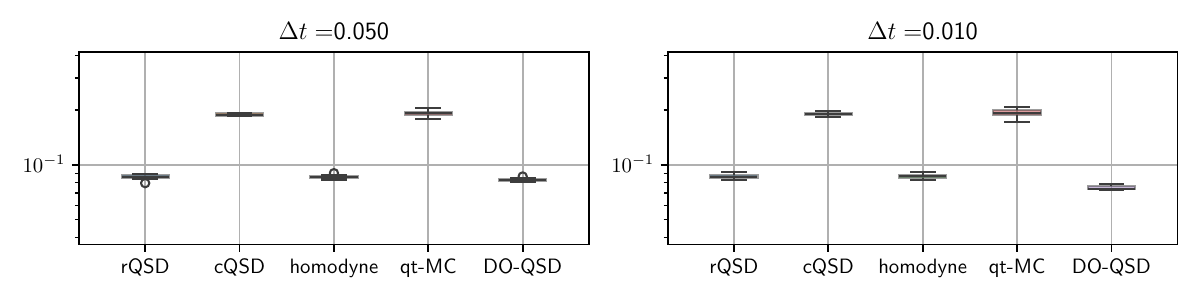}
\caption{Observable $\sigma_Y\otimes \unit_{\text{ph}}$}
\end{subfigure}
\\
\begin{subfigure}[b]{0.9\textwidth}
\includegraphics[width=\textwidth]{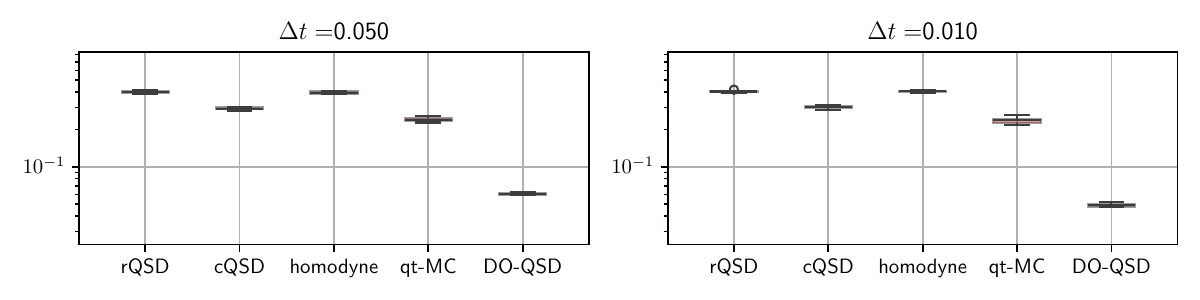}
\caption{Observable $\sigma_Z\otimes \unit_{\text{ph}}$}
\end{subfigure}
\\
\begin{subfigure}[b]{0.9\textwidth}
\includegraphics[width=\textwidth]{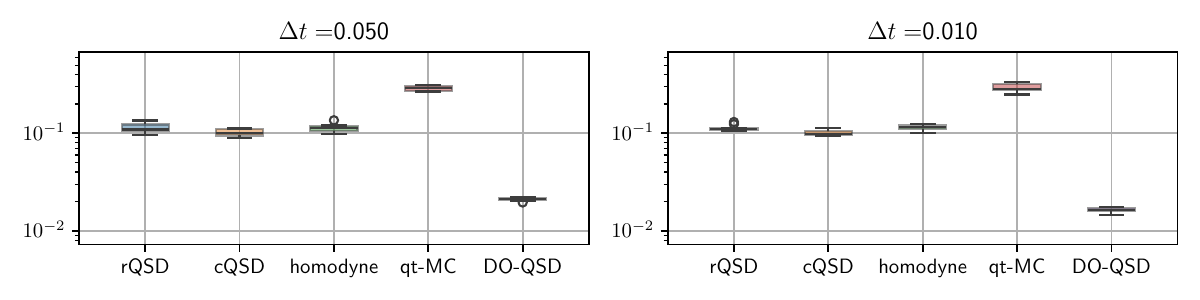}
\caption{Observable $\unit_{\text{atom}} \otimes \hat{a}^\dagger \hat{a}$}
\end{subfigure}
\caption{\textbf{Averaged variance} for various unraveling schemes applied to the atom-cavity model~\eqref{eg::atom::H}, \eqref{eqn::atom::L1}, \eqref{eqn::atom::L4},
for the same experiments as reported in Figure~\ref{fig::eg2}.}
\label{fig::eg2::var}
\end{figure}

\section{Conclusion and outlook}

In this work, we present a complete characterization of stochastic unraveling schemes for a specific scenario, and derive the optimal unraveling scheme for minimizing the rate of change of the classical variance of observables. Here, the classical variance refers to the variance arising from classical computation, rather than the variance from quantum measurement. 
Depending on the problem, the proposed dynamically optimal quantum state diffusion (DO-QSD) may substantially reduce variance, yielding a lower mean-squared error in classical simulations; in general, it performs at least as well as conventional quantum state diffusion.
Furthermore, it can be combined with other sampling techniques, such as control variate \cite{le_bris_adaptive_2015} and multi-level Monte Carlo \cite{giles_multilevel_2015}, among others. We emphasize that the variance reduction in our approach is achieved by selecting an efficient stochastic process, making this method complementary to other sampling strategies rather than serving as an alternative.

Although the DO-QSD may not be globally optimal, the variance reduction that it offers does not require any training, which is typically prohibitively expensive for high-dimensional quantum many-body systems. It remains an open question whether a globally optimal unraveling scheme can be expressed in an explicit form.
The benefits of choosing dynamically optimal unraveling will be more pronounced when the bias from numerical discretization is not significant, which leads to more demand for efficient high-order numerical solvers \cite{steinbach_high-order_1995,kloeden,li2020exponential}. How to choose a better scheme with complexity cost into consideration is not fully clear yet. These questions will be left as future endeavors.

Stochastic unraveling methods are not limited to Markovian quantum dynamics. Non-Markovian dynamics can be described by stochastic differential equations with memory and colored noise \cite{gaspard1999non}, or embedded into a higher-dimensional Markovian stochastic system \cite{strunz1999open,li2021markovian}. Extending these unraveling methods to non-Markovian settings is a promising direction for future work.

\section*{Acknowledgement}

Y.C. was sponsored by Shanghai Pujiang Program (23PJ1404600), NSFC grant No. 12341104, Shanghai Science and Technology Innovation Action Plan (24LZ1401200), and Shanghai Jiao Tong University 2030 Initiative.

\bibliography{ref.bib}

\appendix
\section{Proofs for \secref{sec::sde}}

\subsection{Proof of Theorem~\ref{lem::characterization}}
\label{proof::lem::characterization}

\smallskip
{\noindent {\bf Proof of part (i):} }\\

The main approach is to calculate the matrix elements of \eqref{eqn::unraveling_condition} in the basis formed by $\psi$ and $L\psi$, as well as the orthogonal subspace of $\text{span}\{\psi, L\psi\}$.
To begin with, compared to LQSD, we denote the extra terms in the general ansatz as $A$ and $B$, defined as
\begin{align}
\label{eqn::AB}
\begin{aligned}
A(\psi) &= a(\psi) - \Big( -i {H} - \frac{1}{2} L^\dagger L \Big) \psi ,\qquad 
B(\psi) = b(\psi) - L \psi.
\end{aligned}
\end{align}
From \eqref{eqn::unraveling_condition}, the difference terms $A$ and $B$ must satisfy 
\begin{align}
\label{eqn::new_relation}
\begin{aligned}
& A(\psi) \psi^\dagger + \psi A(\psi)^\dagger  + B(\psi) (L\psi)^\dagger + L\psi B(\psi)^\dagger + B(\psi) B^\dagger(\psi) = 0, \qquad \forall \psi \in \mathcal{M}.
\end{aligned}
\end{align}

Without loss of generality, suppose $\psi \neq 0$, and $\psi$ is not an eigen-vector of $L$.
The reason is that within any compact subspace of $\Complex^n$, the measure of the set of eigen-vectors $\psi$ of $L$ has measure zero; if the measure is not zero, it means that $L$ is proportional to the identity matrix, and the dynamics reduces to von-Neumann equation and the system is basically closed (which is not our setup). Then the following conclusion holds almost everywhere on any compact subset of $\Complex^n$.
In the following, we shall fix such a $\psi$. One can find an orthogonal basis of $\Complex^n$ such that
\begin{align}
\label{eqn::basis}
\phi_1 = \psi, \qquad \phi_2 = L \psi - \frac{\inner{\psi}{L\psi}}{\norm{\psi}^2}\psi,
\end{align}
and all remaining basis vectors can be constructed via the Gram-Schmidt approach. Note that these $\phi_1, \phi_2$ basis vectors would depend on $\psi$ implicitly, but for ease of notation, we will not explicitly indicate that, and the same applies to some other notations below. 

Let $\mathcal{Q}$ be the projection operator to the subspace constructed via $\phi_3, \phi_4, \cdots, \phi_n$. We can decompose $A(\psi)$ and $B(\psi)$ into 
\begin{align*}
A(\psi) = \alpha_1(\psi) \phi_1 + \alpha_2(\psi) \phi_2 + A_{R}(\psi),\\
B(\psi) = \beta_1(\psi) \phi_1 + \beta_2(\psi) \phi_2 + B_R(\psi),
\end{align*}
where $\mathcal{Q} A_R = A_R$ and $\mathcal{Q} B_R = B_R$. By plugging these decompositions into the above equation \eqref{eqn::new_relation}, we have
\begin{align*}
0 =& \big(\alpha_1(\psi) \phi_1 + \alpha_2(\psi) \phi_2 + A_{R}(\psi)\big) \phi_1^\dagger  \\
&\qquad + \phi_1 \big(\alpha_1(\psi)^* \phi_1^\dagger + \alpha_2(\psi)^* \phi_2^\dagger + A_{R}^\dagger(\psi)\big) \\
&\qquad + \big(\beta_1(\psi) \phi_1 + \beta_2(\psi) \phi_2 + B_R(\psi)\big) \big(\phi_2^\dagger + \frac{\inner{\psi}{L\psi}^*}{\norm{\psi}^2} \phi_1^\dagger\big) \\
&\qquad + \big(\phi_2 + \frac{\inner{\psi}{L\psi}}{\norm{\psi}^2} \phi_1\big) \big(\beta_1(\psi)^* \phi_1^\dagger + \beta_2(\psi)^* \phi_2^\dagger + B_R^\dagger(\psi)\big) \\
&\qquad + \big(\beta_1(\psi) \phi_1 + \beta_2(\psi) \phi_2 + B_R(\psi)\big) \big(\beta_1(\psi)^* \phi_1^\dagger + \beta_2(\psi)^* \phi_2^\dagger + B_R^\dagger(\psi)\big).
\end{align*}
By collecting matrix elements in terms of the basis $\big\{\phi_1, \phi_2, \cdots, \phi_n\big\}$, one has
\begin{align*}
\begin{array}{lll}
\inner{\phi_1}{(\cdot) \phi_1} =0 & \implies \qquad & \alpha_1 + \alpha_1^* + \beta_1 \frac{\inner{\psi}{L\psi}^*}{\norm{\psi}^2} + \beta_1^* \frac{\inner{\psi}{L\psi}}{\norm{\psi}^2} + \abs{\beta_1}^2 = 0,\\
\inner{\phi_2}{(\cdot) \phi_1} =0 & \implies \qquad & \alpha_2 + \beta_2 \frac{\inner{\psi}{L\psi}^*}{\norm{\psi}^2} + \beta_1^* + \beta_2 \beta_1^* = 0, \\
\mathcal{Q} {(\cdot) \phi_1} = 0 & \implies \qquad & A_R + B_R \frac{\inner{\psi}{L\psi}^*}{\norm{\psi}^2} + B_R \beta_1^* = 0, \\
\inner{\phi_2}{(\cdot) \phi_2} = 0 & \implies \qquad & \beta_2 + \beta_2^* + \abs{\beta_2}^2 = 0, \\
\mathcal{Q} (\cdot) \phi_2 = 0 & \implies \qquad & B_R + B_R \beta_2^* = 0, \\
\mathcal{Q} (\cdot) \mathcal{Q} = 0 & \implies \qquad & B_R B_R^\dagger = 0.
\end{array}
\end{align*}
In the above, we suppressed the dependence on $\psi$ for ease of notation.

From the sixth line above, we immediately know that $B_R = 0$, and thus the fifth line automatically holds. Then the third line yields $A_R = 0$. As for the fourth line, it can be simplified as follows:
\begin{align*}
\abs{\beta_2(\psi) + 1}^2 = 1.
\end{align*}
By direct simplification, we have
\begin{align}
\label{eqn::relation}
\begin{aligned}
& 2 \re{\alpha_1(\psi)} = - \beta_1(\psi) \frac{\inner{\psi}{L\psi}^*}{\norm{\psi}^2} - \beta_1(\psi)^*  \frac{\inner{\psi}{L\psi}}{\norm{\psi}^2} - \abs{\beta_1(\psi)}^2, \\
& \alpha_2(\psi) = - \beta_2(\psi) \frac{\inner{\psi}{L\psi}^*}{\norm{\psi}^2} - \beta_1(\psi)^* - \beta_2(\psi) \beta_1(\psi)^*, \\
& \abs{\beta_2(\psi) + 1}^2 = 1, \\
& A_R(\psi) = B_R(\psi) = 0.
\end{aligned}
\end{align}
We could simply parameterize $\beta_2(\psi)$ via $-1 + e^{i \theta(\psi)}$, where $\theta$ is a real-valued function.
$\beta_1$ is generally a complex-valued function and the imaginary parts of $\alpha_1$ is denoted as $\gamma$ (which is a real-valued function).

Finally, by combining \eqref{eqn::relation} and \eqref{eqn::basis}, we arrive at the following equivalent class:
\begin{align*}
A(\psi) = &\ \big(- \frac{1}{2} \beta_1(\psi) \frac{\inner{\psi}{L\psi}^*}{\norm{\psi}^2} - \frac{1}{2} \beta_1(\psi)^*  \frac{\inner{\psi}{L\psi}}{\norm{\psi}^2} - \frac{1}{2} \abs{\beta_1(\psi)}^2 + i \gamma(\psi)\big) \psi \\
& \qquad + \big(- (-1+e^{i\theta(\psi)}) \frac{\inner{\psi}{L\psi}^*}{\norm{\psi}^2} - \beta_1(\psi)^* e^{i \theta(\psi)}\big) \big(L \psi - \frac{\inner{\psi}{L\psi}}{\norm{\psi}^2}\psi\big) \\
=&\ \bigg(-\frac{1}{2} \abs{\beta_1(\psi) + \frac{\inner{\psi}{L\psi}}{\norm{\psi}^2}}^2 -\frac{1}{2} \abs{\frac{\inner{\psi}{L\psi}}{\norm{\psi}^2}}^2 + e^{i \theta(\psi)} \frac{\inner{\psi}{L\psi}}{\norm{\psi}^2} \big(\beta_1(\psi)^* + \frac{\inner{\psi}{L\psi}^*}{\norm{\psi}^2}) + i \gamma(\psi) \bigg) \psi  \\
&\qquad +  \big( (1-e^{i\theta(\psi)}) \frac{\inner{\psi}{L\psi}^*}{\norm{\psi}^2} - \beta_1(\psi)^* e^{i \theta(\psi)}\big) \big(L \psi\big), \\
B(\psi) = &\ \beta_1(\psi) \psi + \big(-1 + e^{i \theta(\psi)}\big) \big(L \psi - \frac{\inner{\psi}{L\psi}}{\norm{\psi}^2}\psi\big) \\
=&\ \Big(\beta_1(\psi) + \big(1 - e^{i \theta(\psi)}\big) \frac{\inner{\psi}{L\psi}}{\norm{\psi}^2}\Big)\psi + \big(-1 + e^{i \theta(\psi)}\big) \big(L \psi\big).
\end{align*}
We can observe that it is easier to parameterize $\beta_1(\psi) = \eta(\psi) - \frac{\inner{\psi}{L\psi}}{\norm{\psi}^2}$. Then we can simplify the parameterization via
\begin{align*}
A(\psi) &= \bigg(-\frac{1}{2} \abs{\eta(\psi)}^2 - \frac{1}{2} \abs{\frac{\inner{\psi}{L\psi}}{\norm{\psi}^2}}^2 + e^{i \theta(\psi)} \frac{\inner{\psi}{L\psi}}{\norm{\psi}^2} \eta(\psi)^* + i \gamma(\psi) \bigg) \psi +  \big( \frac{\inner{\psi}{L\psi}^*}{\norm{\psi}^2} -e^{i\theta(\psi)} \eta(\psi)^* \big) \big(L \psi\big)\\
&\qquad = \bigg(-\frac{1}{2} \abs{\eta(\psi) - \frac{e^{i\theta(\psi)}\inner{\psi}{L\psi}}{\norm{\psi}^2}}^2 + i \tilde{\gamma}(\psi) \bigg) \psi +  \big( \frac{\inner{\psi}{L\psi}^*}{\norm{\psi}^2} -e^{i\theta(\psi)} \eta(\psi)^* \big) \big(L \psi\big), \\
B(\psi) &= \Big(\eta(\psi) - e^{i \theta(\psi)}\frac{\inner{\psi}{L\psi}}{\norm{\psi}^2} \Big)\psi + \big(-1 + e^{i \theta(\psi)}\big) \big(L \psi\big),
\end{align*}
where $\tilde{\gamma}(\psi) = \gamma(\psi) + \im{e^{i\theta(\psi)} \frac{\inner{\varphi}{L\varphi}}{\norm{\varphi}^2} \eta(\psi)^*}$ is still an arbitrary real-valued function. 
Since there is certain repetitive pattern in $\eta$, we can introduce $\tilde{\eta}(\psi) = \eta(\psi) - e^{i \theta(\psi)}\frac{\inner{\psi}{L\psi}}{\norm{\psi}^2}$, and then we end up with 
\begin{align*}
A(\psi) &= \bigg(-\frac{1}{2} \abs{\tilde{\eta}(\psi)}^2 + i \tilde{\gamma}(\psi) \bigg) \psi - e^{i \theta(\psi)} \tilde{\eta}(\psi)^* \big(L \psi\big)\\
B(\psi) &= \tilde{\eta}(\psi)\ \psi + \big(-1 + e^{i \theta(\psi)}\big) \big(L \psi\big).
\end{align*}

By returning to the original drift and diffusion terms using \eqref{eqn::AB}, we obtain 
\begin{align*}
\begin{aligned}
a(\psi) &= -i {H}\psi - \frac{1}{2} L^\dagger L\psi + \bigg(-\frac{1}{2} \abs{\tilde{\eta}(\psi)}^2 + i \tilde{\gamma}(\psi) \bigg) \psi - e^{i \theta(\psi)} \tilde{\eta}(\psi)^* \big(L \psi\big), \\
b(\psi) &= \tilde{\eta}(\psi)\ \psi + e^{i \theta(\psi)} \big(L \psi\big).
\end{aligned}
\end{align*}
By dropping the notation of tilde, we end up with the simplified expression in \eqref{eqn::ab_Cn}.

\bigskip
{\noindent {\bf Proof of part (ii):}}\\

We will verify directly that the norm is always preserved in the above choice \eqref{eqn::ab_Cn} when restricted to the space of the unit sphere $\mathcal{S}^n$.
Let us consider the SDE \eqref{eqn::sde_1} with $\norm{\psi(t)}^2 = 1$ at time $t$. We want to validate that the norm won't change:
\begin{align*}
\ud \inner{\psi}{\psi} =&\ \innerbig{\psi}{-i {H}\psi - \frac{1}{2} L^\dagger L\psi + \bigg(-\frac{1}{2} \abs{{\eta(\psi)}}^2 + i {\gamma}(\psi) \bigg) \psi - e^{i \theta(\psi)} \eta(\psi)^* L \psi}\ \ud t + c.c.  \\
&\qquad + \innerbig{\psi}{\eta(\psi) \psi + e^{i \theta(\psi)} L\psi}\ \ud W + c.c.  \\
&\qquad + \innerbig{\eta(\psi)\ \psi + e^{i \theta(\psi)} L \psi}{\eta(\psi)\ \psi + e^{i \theta(\psi)} L \psi}\ \ud t \\
=&\ \Big(- \norm{L\psi}^2 - \abs{\eta(\psi)}^2 - e^{i\theta(\psi)} \eta(\psi)^* \inner{\psi}{L\psi} - e^{-i\theta(\psi)} \eta(\psi) \inner{L\psi}{\psi} \\
&\qquad + \abs{\eta(\psi)}^2 + \eta(\psi)^* e^{i \theta(\psi)} \inner{\psi}{L\psi} + \eta(\psi) e^{-i\theta(\psi)}\inner{L\psi}{\psi} + \norm{L\psi}^2\Big)\ \ud t\\
& \qquad + \big(\eta(\psi) + \eta(\psi)^* + e^{i\theta(\psi)} \inner{\psi}{L\psi} + e^{-i\theta(\psi)} \inner{L\psi}{\psi}\big) \ud W \\
=&\ \big(\eta(\psi) + \eta(\psi)^* + e^{i\theta(\psi)} \inner{\psi}{L\psi} + e^{-i\theta(\psi)} \inner{L\psi}{\psi}\big) \ud W.
\end{align*}
We suppressed the time-dependence $\psi = \psi(t)$ for simplicity; the term $c.c.$ means complex conjugate of the previous term. The stochastic fluctuation won't affect it iff
$
\eta(\psi) + e^{i\theta(\psi)} \inner{\psi}{L\psi}
$
is purely imaginary, and thus we can parameterize it via
\begin{align*}
\eta(\psi) = i h(\psi) - e^{i\theta(\psi)} \inner{\psi}{L\psi},
\end{align*}
where $h$ is an arbitrary real-valued function.

\subsection{Proof of Theorem~\ref{thm::optimal_soln}}
\label{proof::equiv_opti}

We first prove the main theorem based on Lemma~\ref{lem::equiv_opti} and then we will come back to prove this lemma.\\

{\noindent {\textbf{Proof of Theorem~\ref{thm::optimal_soln}}:}} \\

Due to the linearity, we only need to consider the form \eqref{eqn::sde_1}. 
\begin{lemma}
\label{lem::equiv_opti}
For the second-order moment, one has
\begin{align*}
&\ \frac{\ud}{\ud t} \ee \abs{\inner{\psi(t)}{\obs\psi(t)}}^2 
=\ 2\ \ee \Big[\innerbig{\obs \psi}{\mathcal{L}(\psi\psi^\dagger) \obs \psi}\Big]  + 2\ \ee\Big[\inner{L\psi}{\obs L\psi}\inner{\psi}{\obs\psi} \Big]  - 2\ \ee\Big[\abs{\inner{\obs\psi}{L\psi}}^2\Big]  + 4\ \loss_1,
\end{align*}
where $\psi = \psi(t)$, and the loss function $\loss_1$ is defined as
\begin{align}
\label{eqn::loss1}
\begin{aligned}
\loss_1 :=&\ \ee\ \absBig{\reBig{\eta(\psi) \inner{\psi}{\obs\psi} + e^{i \theta(\psi)} \inner{\obs \psi}{L\psi}}}^2.
\end{aligned}
\end{align}
Therefore, the above optimization problem \eqref{eqn::deri_var} is equivalent to minimize the loss $\loss_1$.

If we further additional impose the norm-preservation \eqref{eqn::norm}, then $\re{\eta(\psi)} =  - \re{e^{i\theta(\psi)} \inner{\psi}{L\psi}}$ and the loss function becomes
\begin{align}
\begin{aligned}
\loss_2 :=&\  \ee\ \absBig{\reBig{e^{i \theta(\psi)}  \big(\inner{\obs \psi}{L\psi} - \inner{\psi}{L\psi} \inner{\psi}{\obs\psi})}}^2.
\end{aligned}
\end{align}
\end{lemma}

The proof is postponed to the later part of this section. 
As one could observe, the term $\gamma(\psi)$ does not show up in the loss function. If we further impose \eqref{eqn::norm}, since $\re{\eta(\psi)} =  - \re{e^{i\theta(\psi)} \inner{\psi}{L\psi}}$ does not involve $h(\psi)$, then the only variable for optimization is simply $\theta(\psi)$.

The optimization problem \eqref{eqn::deri_var} admits the following explicit global minimizers:
\begin{itemize}
\item If we impose norm-preservation \eqref{eqn::norm} a priori, and denote 
\begin{align*}
\inner{\obs \psi}{L\psi} - \inner{\psi}{L\psi} \inner{\psi}{\obs\psi} =  R(\psi) e^{i \phase(\psi)},
\end{align*}
where $R(\psi)$ is the norm and $\phase(\psi)$ is the phase factor,  
then the optimal $\theta^\star$ is given by 
$
\theta^\star(\psi) = - \phase(\psi) + \frac{\pi(1 + 2 m)}{2},$ where $m \in \mathbb{Z}.
$
Namely, 
\begin{align}
\label{eqn::theta_optimal}
e^{i \theta^\star(\psi)} = \pm i e^{-i \phase(\psi)}.
\end{align}
For this case, $\loss_2 = 0$ so that the global minimum is achieved.

\item For the general case, for any $\theta(\psi)$, we can pick $\eta(\psi)$ so that 
$\eta(\psi) \inner{\psi}{\obs\psi} + e^{i \theta(\psi)}  \inner{\obs \psi}{L\psi}$
is purely imaginary, so that it achieves $\loss_1 = 0$.
\end{itemize}

Then the conclusion of Theorem~\ref{thm::optimal_soln} easily follows by the linearity.\\

\noindent \textbf{Proof of Lemma~\ref{lem::equiv_opti}:}\\

For simplicity, we will suppress the time dependence $\psi = \psi(t)$ for clarity.
By Itô calculus and some direct simplifications, we have
\begin{align}
\label{eqn::qsd_derivative}
\begin{aligned}
&\ \frac{\ud}{\ud t}\ \ee\ \absbig{\innerbig{\psi(t)}{\obs\psi(t)}}^2\\
 =&\ 2 \ee \Big[\inner{\psi}{\obs a(\psi)} \inner{\psi}{\obs \psi}\Big] + c.c.\\
&\qquad + \ee \Big[\innerbig{b(\psi)}{\obs b(\psi)}\inner{\psi}{\obs \psi} + \innerbig{b(\psi)}{\obs \psi}\inner{b(\psi)}{\obs \psi} + \innerbig{b(\psi)}{\obs \psi}\inner{\psi}{\obs b(\psi)}\Big] + c.c.\\
\myeq{\eqref{eqn::unraveling_condition}}&\ 2 \ee \Big[ \innerbig{\obs \psi}{\big(\mathcal{L}(\psi\psi^\dagger) - b(\psi) b(\psi)^\dagger\big) \obs \psi}\Big] \\
&\qquad + \Big( \ee \Big[ \innerbig{b(\psi)}{\obs b(\psi)}\inner{\psi}{\obs \psi}+\innerbig{b(\psi)}{\obs \psi}\inner{b(\psi)}{\obs \psi} + \innerbig{b(\psi)}{\obs \psi}\inner{\psi}{\obs b(\psi)}\Big] + c.c. \Big)\\
=&\ 2 \ee \Big[\innerbig{\obs \psi}{\mathcal{L}(\psi\psi^\dagger) \obs \psi}\Big] \\
&\qquad + \Big( \ee \Big[\innerbig{b(\psi)}{\obs b(\psi)}\inner{\psi}{\obs \psi}+\innerbig{b(\psi)}{\obs \psi}\inner{b(\psi)}{\obs \psi}\Big] + c.c. \Big)\\
\end{aligned}
\end{align}
The first term above is independent of functions $a$ and $b$, so that this term is intrinsic due to the choice of the Lindblad equation and the observable.
The remaining terms of the loss function are
\begin{align*}
&\ \ee \Big[\inner{b(\psi)}{\obs b(\psi)}\inner{\psi}{\obs \psi} + \inner{b(\psi)}{\obs \psi}\inner{b(\psi)}{\obs \psi} \Big] + c.c. \\
 =&\ \ee \Big[\innerbig{\eta(\psi) \psi + e^{i \theta(\psi)} L \psi
}{\obs \big(\eta(\psi) \psi + e^{i \theta(\psi)} L \psi
\big)}\inner{\psi}{\obs \psi} \Big] + c.c. \\
&\qquad + \ee \Big[\inner{\eta(\psi) \psi + e^{i \theta(\psi)} L \psi
}{\obs \psi}\inner{\eta(\psi) \psi + e^{i \theta(\psi)} L \psi
}{\obs \psi}\Big] + c.c. \\
=&\ \ee\Big[ 2 \abs{\eta(\psi)}^2 \inner{\psi}{\obs\psi}^2 \Big] \\
&\qquad + \ee\Big[ 2 \eta^*(\psi)e^{i \theta(\psi)}  \inner{\psi}{\obs L\psi} \inner{\psi}{\obs\psi} + 2 \eta(\psi)e^{-i\theta(\psi)}\inner{L\psi}{\obs \psi} \inner{\psi}{\obs\psi}\Big] \\
&\qquad + 2 \ee\Big[\inner{L\psi}{\obs L\psi}\inner{\psi}{\obs\psi} \Big] 
+ \ee\Big[ (\eta(\psi)^2 + (\eta(\psi)^*)^2 ) \inner{\psi}{\obs\psi}^2 \Big] \\
&\qquad + \ee\Big[ 2 \eta(\psi)^* e^{-i \theta(\psi)} \inner{\psi}{\obs\psi} \inner{L\psi}{\obs\psi} + 2 \eta(\psi) e^{i\theta(\psi)} \inner{\psi}{\obs\psi} \inner{\psi}{\obs L\psi} \Big] \\
&\qquad + \ee \Big[ e^{-2 i \theta(\psi)} \inner{L\psi}{\obs\psi}^2 + e^{2i\theta(\psi)}\inner{\psi}{\obs L\psi}^2 \Big] \\
=&\ \ee\Big[ 4 \re{\eta(\psi)}^2 \inner{\psi}{\obs\psi}^2 \Big] \\
&\qquad + \ee\Big(4 \re{\eta(\psi)} e^{i \theta(\psi)} \inner{\psi}{\obs L\psi} \inner{\psi}{\obs\psi} + 4 \re{\eta(\psi)} e^{-i\theta(\psi)}\inner{L\psi}{\obs \psi} \inner{\psi}{\obs\psi}\Big) \\
&\qquad + 2 \ee\Big[ \inner{L\psi}{\obs L\psi}\inner{\psi}{\obs\psi} \Big]  + \ee\Big[ e^{-2 i \theta(\psi)} \inner{L\psi}{\obs\psi}^2 + e^{2i\theta(\psi)}\inner{\psi}{\obs L\psi}^2 \Big].
\end{align*}
Since the second last term also does not involve functions $a$, $b$, it is equivalent to minimizing
\begin{align*}
& \ee\Big[ 4 \re{\eta(\psi)}^2 \inner{\psi}{\obs\psi}^2 \Big] \\
&\qquad + \ee\Big[ 4 \re{\eta(\psi)} e^{i \theta(\psi)}  \inner{\psi}{\obs L\psi} \inner{\psi}{\obs\psi} + 4 \re{\eta(\psi)} e^{-i\theta(\psi)}\inner{L\psi}{\obs \psi} \inner{\psi}{\obs\psi}\Big] \\
&\qquad + \ee\Big[ e^{-2 i \theta(\psi)} \inner{L\psi}{\obs\psi}^2 + e^{2i\theta(\psi)}\inner{\psi}{\obs L\psi}^2 \Big] \\
=&\ \ee\Big[ 4 \re{\eta(\psi)}^2 \inner{\psi}{\obs\psi}^2 \Big] 
+ \ee\Big[ 8 \re{\eta(\psi)} \rebig{e^{i \theta(\psi)}  \inner{\obs \psi}{L\psi}} \inner{\psi}{\obs\psi} \Big] \\
&\qquad + \ee\Big[ 2 \re{e^{i\theta(\psi)} \inner{\obs\psi}{L\psi}}^2 - 2 \im{e^{i\theta(\psi)} \inner{\obs\psi}{L\psi}}^2 \Big] \\
=&\ 4\ \ee\Big\lvert\re{\eta(\psi)} \inner{\psi}{\obs\psi} + \re{e^{i \theta(\psi)}  \inner{\obs \psi}{L\psi}}\Big\rvert^2 \\
&\qquad - 2\ \ee\Big[ \re{e^{i\theta(\psi)} \inner{\obs\psi}{L\psi}}^2 + \im{e^{i\theta(\psi)} \inner{\obs\psi}{L\psi}}^2\Big] \\
=&\ 4\ \ee\Big\lvert \re{\eta(\psi)} \inner{\psi}{\obs\psi} + \re{e^{i \theta(\psi)}  \inner{\obs \psi}{L\psi}}\Big\rvert^2  - 2\ \ee\Big[\abs{e^{i\theta(\psi)} \inner{\obs\psi}{L\psi}}^2\Big] \\
=&\ 4\ \ee\Big\vert\re{\eta(\psi)} \inner{\psi}{\obs\psi} + \re{e^{i \theta(\psi)}  \inner{\obs \psi}{L\psi}}\Big\rvert^2  - 2\ \ee\Big[\abs{\inner{\obs\psi}{L\psi}}^2\Big].
\end{align*}
Note that the last term in the last line also does not involve functions $a$ and $b$, the loss function is simply
\begin{align*}
\loss_1 = \ee\Big\lvert\re{\eta(\psi)} \inner{\psi}{\obs\psi} + \re{e^{i \theta(\psi)}  \inner{\obs \psi}{L\psi}}\Big\rvert^2.
\end{align*}

In summary,
\begin{align*}
\frac{\ud}{\ud t} \ee \abs{\inner{\psi(t)}{\obs\psi(t)}}^2 =&\ 2 \ee \Big[\innerbig{\obs \psi}{\mathcal{L}(\psi\psi^\dagger)\obs \psi}\Big]  + 2 \ee\Big[\inner{L\psi}{\obs L\psi}\inner{\psi}{\obs\psi} \Big] - 2\ \ee\Big[\abs{\inner{\obs\psi}{L\psi}}^2\Big] + 4 \loss_1,
\end{align*}
where all previous terms before $\loss_1$ are independent of the unraveling scheme.
When we impose the norm preservation, the conclusion easily follows from \eqref{eqn::norm}.

\section{Proofs for \secref{sec::sjp}}
\subsection{Proof of Lemma~\ref{lem::sjp_unraveling}}
\label{proof::lem::sjp}

The conditions for the norm-preservation is straightforward to validate and the detailed proof is thus skipped. 
Next we will validate the unraveling constraints. Suppose the state is $\psi(t)$ at time $t$. Then locally in time $[t, t+\dt]$, the state will evolve to $\psi(t+\dt)$ according to the following jump process \eqref{eqn::jump_numerical}, and one has 
\begin{align*}
\ee\big[\psi(t+\dt) \psi(t+\dt)^\dagger\big] =&\ \Big(1 - \sum_{j} \jumprate_j\big(\psi(t)\big) \dt\Big) \Big(\psi(t) + a\big(\psi(t)\big)\dt \Big) \Big(\psi(t) + a\big(\psi(t)\big)\dt \Big)^\dagger \\
& + \sum_{j} \jumprate_j\big(\psi(t)\big) \dt \Big(\psi(t) + b_j\big(\psi(t)\big)\Big) \Big(\psi(t) + b_j\big(\psi(t)\big)\Big)^\dagger + \mathcal{O}\big(\dt^2\big) \\
=&\ \psi(t) \psi(t)^\dagger + \dt\Big(\psi(t)  a\big(\psi(t)\big)^\dagger + a\big(\psi(t)\big) \psi(t)^\dagger\Big) \\
& + \sum_{j} \dt \jumprate_j\big(\psi(t)\big) \Big(\psi(t) b_j\big(\psi(t)\big)^\dagger + b_j\big(\psi(t)\big) \psi(t)^\dagger  + b_j\big(\psi(t)\big) b_j\big(\psi(t)\big)^\dagger \Big) + \mathcal{O}\big(\dt^2\big).
\end{align*}
By matching this with the Lindblad equation, we immediately have \eqref{eqn::qjp_condition}.

\subsection{Proof of Theorem~\ref{thm::qjp_characterize}}
\label{proof::qjp_characterize}

\medskip

Assuming stochastic unraveling, we aim to establish \eqref{eqn::opti_jump}. 
Due to the similarity in the structure of the SDE and jump process cases in 
\eqref{eqn::unraveling_condition} and \eqref{eqn::qjp_condition}, we observe that for any $\psi\in \mathcal{S}^n \backslash \mathcal{E}$ where $\mathcal{E}$ is the set of eigenvectors of $L$, 
\begin{align*}
\sqrt{\jumprate(\psi)} b(\psi) &= \eta(\psi)\ \psi + e^{i \theta(\psi)} L \psi, \\
a(\psi) + \jumprate(\psi) b(\psi) &= -i {H}\psi - \frac{1}{2} L^\dagger L\psi 
+ \big(-\frac{1}{2} \abs{{\eta}(\psi)}^2 
+ i {\gamma(\psi)} \big) \psi - e^{i \theta(\psi)} \eta(\psi)^* L \psi,
\end{align*}
where $\eta$ is some complex-valued function, and $\theta, \gamma$ are real-valued functions.\\

For any $\psi\in \mathcal{S}^n\backslash \mathcal{E}$, depending on the value of $\jumprate(\psi)$, we consider two cases.\\

{\noindent \emph{Case 1:}} When $\jumprate(\psi) = 0$, we necessarily need to impose 
\begin{align*}
\eta(\psi) \psi + e^{i \theta(\psi)} L \psi = 0 \implies L \psi = -e^{-i\theta(\psi)}\eta(\psi) \psi,
\end{align*}
which means $\psi$ is necessarily an eigenvector of $L$. 
In the space $\mathcal{S}^n$, the set of such eigenvectors has measure zero. 
Therefore, this constraint is not the most important one. Besides, it contradicts with $\psi\notin \mathcal{E}$. Therefore, we don't need to consider this case.
\bigskip

{\noindent \emph{Case 2:}} When $\jumprate(\psi) > 0$, then functions $a$ and $b$ can be expressed as
\begin{align*}
b(\psi) &= \frac{\eta(\psi)}{\sqrt{\jumprate(\psi)}} \psi + \frac{e^{i\theta(\psi)}}{\sqrt{\jumprate(\psi)}} L \psi, \\
a(\psi) &=  -i {H}\psi - \frac{1}{2} L^\dagger L\psi + \big(-\frac{1}{2} \abs{{\eta(\psi)}}^2 
+ i {\gamma(\psi)} \big) \psi - e^{i \theta(\psi)} \eta^*(\psi) L \psi - \sqrt{\jumprate(\psi)} \big(\eta(\psi) \psi + e^{i \theta(\psi)} L \psi\big),
\end{align*}
are uniquely determined by $\eta$, $\theta$, $\gamma$, and $\jumprate$. To fulfill the norm-preservation constraints for $b$, we have
\begin{align*}
\Big\Vert(1 + \frac{\eta(\psi)}{\sqrt{\jumprate(\psi)}}) \psi 
+ \frac{e^{i\theta(\psi)}}{\sqrt{\jumprate(\psi)}} L \psi \Big\Vert = 1, 
\qquad \forall \psi \in \mathcal{S}^n\backslash \mathcal{E}.
\end{align*}
By directly computing the square of $\ell_2$ norm and performing direct expansion, 
we know the following by collecting terms related to $\jumprate$,
\begin{align}
\label{eqn::norm_a_b}
\begin{aligned}
& \sqrt{\jumprate(\psi)} \big(\eta(\psi) + \eta(\psi)^* + e^{i \theta(\psi)} \inneravg{L}{\psi} + e^{-i\theta(\psi)} \inneravg{L^\dagger}{\psi}) \\
& \qquad + \abs{\eta(\psi)}^2 + \eta(\psi)^* e^{i \theta(\psi)} \inneravg{L}{\psi} 
+ \eta(\psi) e^{-i \theta(\psi)} \inneravg{L^\dagger}{\psi} + \inneravg{L^\dagger L}{\psi} = 0.
\end{aligned}
\end{align}
By considering the norm-preservation constraint for $a$, we could validate that the constraint remains the same as the last equation \eqref{eqn::norm_a_b}. By expansion, 
\begin{align*}
& \sqrt{\jumprate(\psi)} \Big(
    2 \re{\eta(\psi)} + 2 \re{e^{i\theta(\psi)} \inneravg{L}{\psi}}\Big) 
    + \re{\eta(\psi)}^2 + \im{\eta(\psi)}^2 \\
    & \qquad + 2 \re{\eta(\psi)} \re{e^{i \theta(\psi)} \inneravg{L}{\psi}} 
    + 2 \im{\eta(\psi)} \im{e^{i\theta(\psi)} \inneravg{L}{\psi}} + \inneravg{L^\dagger L}{\psi} = 0.
\end{align*}
By completing the square,
\begin{align*}
& \Big(\re{\eta(\psi)} + \re{e^{i\theta(\psi)} \inneravg{L}{\psi}} + \sqrt{\jumprate(\psi)}\Big)^2 \\
&\qquad + \Big(\im{\eta(\psi)} 
+ \im{e^{i\theta(\psi)} \inneravg{L}{\psi}}\Big)^2 - \abs{e^{i\theta(\psi)} \inneravg{L}{\psi}}^2 - \jumprate(\psi) + \inneravg{L^\dagger L}{\psi} = 0.
\end{align*}
This means 
\begin{align*}
\absbig{\eta(\psi) + e^{i \theta(\psi)} \inneravg{L}{\psi} + \sqrt{\jumprate(\psi)}}^2  
= \jumprate(\psi) + \abs{\inneravg{L}{\psi}}^2 - \inneravg{L^\dagger L}{\psi}.
\end{align*}
This will necessarily impose 
\begin{align*}
& \jumprate(\psi) + \abs{\inneravg{L}{\psi}}^2 - \inneravg{L^\dagger L}{\psi} \ge 0, \\
& \eta(\psi) = e^{i \beta(\psi)} \sqrt{\jumprate(\psi) + \abs{\inneravg{L}{\psi}}^2 - \inneravg{L^\dagger L}{\psi}} - e^{i\theta(\psi)} \inneravg{L}{\psi} - \sqrt{\jumprate(\psi)}.
\end{align*}
where $\beta$ is an arbitrary real-valued function.

In summary, for any $\psi\in \mathcal{S}^n\backslash \mathcal{E}$, we have $\jumprate(\psi) >0$; besides, one has $\jumprate(\psi) + \abs{\inneravg{L}{\psi}}^2 - \inneravg{L^\dagger L}{\psi} \ge 0$
and 
\begin{align*}
\left\{
\begin{aligned}
a(\psi) &=  -i {H}\psi - \frac{1}{2} L^\dagger L\psi + \big(-\frac{1}{2} \abs{{\eta}(\psi)}^2 + i {\gamma(\psi)} \big) \psi \\
&\qquad - e^{i \theta(\psi)} \eta(\psi)^* L \psi - \sqrt{\jumprate(\psi)} \big(\eta(\psi)\ \psi + e^{i \theta(\psi)} L \psi\big), \\
b(\psi) &= \frac{\eta(\psi)}{\sqrt{\jumprate(\psi)}} \psi + \frac{e^{i\theta(\psi)}}{\sqrt{\jumprate(\psi)}} L \psi, \\
\eta(\psi) &= e^{i \beta(\psi)} \sqrt{\jumprate(\psi) + \abs{\inneravg{L}{\psi}}^2 
- \inneravg{L^\dagger L}{\psi}} - e^{i\theta(\psi)} \inneravg{L}{\psi} - \sqrt{\jumprate(\psi)},
\end{aligned}
\right.
\end{align*}
where $\theta$, $\gamma$, $\beta$ are arbitrary real-valued functions.

To simplify the statement about $\jumprate$, it is not hard to show that it is equivalent to the existence of a non-negative function $\alpha$ such that 
\begin{align*}
\jumprate(\psi) = \alpha(\psi) - \abs{\inneravg{L}{\psi}}^2 + \inneravg{L^\dagger L}{\psi}.
\end{align*}
In this way, $\jumprate$ is automatically non-negative. This completes the proof.

\subsection{Proof of Theorem~\ref{thm::optimal_soln_jump}}
\label{proof::optimal_soln_jump}

\smallskip
{\noindent \textbf{Proof of Part (i):}}\\

We will similarly consider the case $K = 1$ (namely, one Lindblad operator only) for clarity of notations, 
since the conclusion can be generalized to arbitrary $K$ simply due to additivity. 
By the jump process \eqref{eqn::sjp_1}, it is straightforward to compute that
\begin{align*}
&\ \frac{\ud}{\ud t} \ee \abs{\inner{\psi(t)}{\obs\psi(t)}}^2\\
 =&\ \ee\Big[ 2\innerbig{a(\psi)}{O\psi}\inner{\psi}{O\psi} + 2\innerbig{\psi}{O a(\psi)}\inner{\psi}{O\psi} - \jumprate(\psi) \inner{\psi}{O\psi}^2 + \jumprate(\psi)\innerbig{b(\psi)+\psi}{O(b(\psi)+\psi)}^2\ \Big]\\
\myeq{\eqref{eqn::qjp_condition}}&\  \ee\Big[2 \innerbig{O\psi}{\mathcal{L}(\psi\psi^\dagger) O \psi}\Big] \\
&\qquad + \ee\Big[ - 2 \jumprate(\psi) \innerbig{\psi}{O b(\psi)}\inner{\psi}{O\psi} - 2 \jumprate(\psi) \inner{\psi}{O\psi}\innerbig{b(\psi)}{O\psi} - 2 \jumprate(\psi) \innerbig{\psi}{O b(\psi)}\innerbig{b(\psi)}{O\psi}\Big] \\
&\qquad + \ee\Big[ - \jumprate(\psi) \inner{\psi}{O\psi}^2\Big]
+ \ee\Big[\jumprate(\psi) \innerbig{b(\psi)+\psi}{O(b(\psi)+\psi)}^2\ \Big].
\end{align*}
We again notice that $\gamma$ does not play any role in the task of variance reduction, and it is simply a redundant parameter; 
without loss of generality, we could let $\gamma = 0$.

Since the first term does not depend on tunable parameters $a$ and $b$, it is enough to consider the term 
\begin{align*}
\mathscr{T} =& - 2 \jumprate(\psi) \innerbig{\psi}{O b(\psi)}\inner{\psi}{O\psi} 
- 2 \jumprate(\psi) \inner{\psi}{O\psi}\innerbig{b(\psi)}{O\psi} 
- 2 \jumprate(\psi) \innerbig{\psi}{O b(\psi)}\innerbig{b(\psi)}{O\psi} \\
& \qquad - \jumprate(\psi) \inner{\psi}{O\psi}^2 
+ \jumprate(\psi) \innerbig{b(\psi)+\psi}{O (b(\psi) +\psi)}^2 \\
=& \jumprate(\psi) \left(
\begin{aligned}
& \inner{\psi}{O b(\psi)}^2 + \innerbig{b(\psi)}{O\psi}^2 
+ \innerbig{b(\psi)}{O b(\psi)}^2 \\
& \qquad + 2 \inner{\psi}{O \psi} \inner{b(\psi)}{O b(\psi)} 
+ 2 \inner{\psi}{O b(\psi)} \innerbig{b(\psi)}{O b(\psi)} \\
&\qquad + 2 \innerbig{b(\psi)}{O \psi} \innerbig{b(\psi)}{O b(\psi)} 
\end{aligned}\right) \\
=&  \underbrace{\jumprate(\psi) \Big(\inner{\psi}{O b(\psi)}^2 + \innerbig{b(\psi)}{\obs \psi}^2 
+ 2 \inner{\psi}{\obs\psi} \innerbig{b(\psi)}{\obs b(\psi)} \Big)}_{\mathscr{T}_1} \\
&\qquad + 
\underbrace{\jumprate(\psi) \Big(\innerbig{b(\psi)}{\obs b(\psi)}^2 
+ 2 \innerbig{\psi}{\obs b(\psi)} \innerbig{b(\psi)}{\obs b(\psi)} 
+ 2 \innerbig{b(\psi)}{\obs \psi} \innerbig{b(\psi)}{\obs b(\psi)} \Big)}_{\mathscr{T}_2}.
\end{align*}
If we consider the first grouped term $\mathscr{T}_1$ above, we notice that it also shows up in the quantum state diffusion (QSD) case; e.g., see the last line in \eqref{eqn::qsd_derivative}. We have already computed in the diffusion case that
\begin{align*}
\mathscr{T}_1 = &\ 2 \inner{L\psi}{\obs L\psi}\inner{\psi}{\obs\psi} 
- 2 \abs{\inner{\obs\psi}{L\psi}}^2 
+ 4 \Big(\re{\eta(\psi)} \inner{\psi}{\obs\psi} 
+ \re{e^{i \theta(\psi)} \inner{\obs \psi}{L\psi}} \Big)^2.
\end{align*}
Therefore, we can simplify the time derivative via
\begin{align}
\label{eqn::deri_sec::1}
\begin{aligned}
 \frac{\ud}{\ud t} \ee\ \abs{\inner{\psi(t)}{\obs\psi(t)}}^2 =&\ 2\ee\Big[\inner{O\psi}{\mathcal{L}(\psi\psi^\dagger) O \psi}\Big] + 2 \ee\Big[\inner{L\psi}{\obs L\psi}\inner{\psi}{\obs\psi}\Big] - 2 \ee\Big[\abs{\inner{\obs\psi}{L\psi}}^2\Big] \\
 & \qquad + 4 \ee\Big[\Big(\re{\eta(\psi)} \inner{\psi}{\obs\psi} + \re{e^{i \theta(\psi)} \inner{\obs \psi}{L\psi}} \Big)^2\Big] \\
 &\qquad + \ee\big[\mathscr{T}_2\big].
 \end{aligned}
\end{align}
Compared to Lemma~\ref{lem::equiv_opti}, 
the difference here is that the jump process contains additional third and fourth-order terms with respect to $b$, namely terms inside $\mathscr{T}_2$. Due to the multiple appearances of the following term, denote 
\begin{align}
\label{eqn::Z}
Z(\psi) = e^{i \theta(\psi)} \inner{\obs \psi}{L\psi}.
\end{align}
Recall the expression of $b$ from \eqref{eqn::opti_jump}, and by direct calculation, one could readily show that 
\begin{align}
\label{eqn::bOb}
\begin{aligned}
\innerbig{b(\psi)}{\obs b(\psi)} &= \frac{1}{\jumprate(\psi)} \Big(\abs{\eta(\psi)}^2 \inneravg{\obs}{\psi} 
+ 2 \re{\eta(\psi)} \re{Z(\psi)} + 2 \im{\eta(\psi)} \im{Z(\psi)} + \inneravg{L^\dagger O L}{\psi} \Big), \\
\innerbig{\psi}{\obs b(\psi)} &= \frac{1}{\sqrt{\jumprate(\psi)}} \big(\eta(\psi) \inneravg{O}{\psi} + Z(\psi)\big),
\end{aligned}
\end{align}
and hence, by combing \eqref{eqn::deri_sec::1}, \eqref{eqn::bOb}, and the definition of $\mathscr{T}_2$ above,
\begin{align*}
&\  \frac{\ud}{\ud t} \ee\ \abs{\inner{\psi(t)}{\obs\psi(t)}}^2 \\
=&\ 2\ee\Big[\inner{O\psi}{\mathcal{L}(\psi\psi^\dagger) O \psi}\Big]
+ 2 \ee\Big[\inner{L\psi}{\obs L\psi}\inner{\psi}{\obs\psi}\Big]
- 2 \ee\Big[\abs{\inner{\obs\psi}{L\psi}}^2\Big] \\
&\qquad +  4 \ee\Big[\Big(\re{\eta(\psi)} \inneravg{\obs}{\psi} + \re{Z(\psi)}\Big)^2\Big]\\
&\qquad + \ee\bigg[\jumprate(\psi) \bigg(
    \innerbig{b(\psi)}{\obs b(\psi)}^2 
    + \frac{4}{\sqrt{\jumprate(\psi)}} 
    \Big(\re{\eta(\psi)}\inneravg{\obs}{\psi} 
    + \re{Z(\psi)}\Big) 
    \innerbig{b(\psi)}{\obs b(\psi)}\bigg)\bigg] \\
=&\ 2\ee\Big[\inner{O\psi}{\mathcal{L}(\psi\psi^\dagger) O \psi}\Big] 
+ 2 \ee\Big[\inner{L\psi}{\obs L\psi}\inner{\psi}{\obs\psi}\Big] 
- 2 \ee\Big[\abs{\inner{\obs\psi}{L\psi}}^2\Big] \\
&\qquad +  \ee\bigg[\jumprate(\psi)\Big(\innerbig{b(\psi)}{\obs b(\psi)} 
+ \frac{2}{\sqrt{\jumprate(\psi)}} \big(\re{\eta(\psi)} \inneravg{\obs}{\psi} 
+ \re{Z(\psi)}\big) \Big)^2 \bigg].
\end{align*}
Since the last term is non-negative, by combining the above with \eqref{eqn::p_optimal},
\begin{align*}
\mathsf{dV}_{\text{diffusion}} \le \mathsf{dV}_{\text{jump}}.
\end{align*}
As mentioned above, the above inequality for the case $K >1$ can be similarly proved via the additivity of the Lindblad operators. Even though locally the jump process won't be superior to the diffusion process due to the last inequality,
we will still show that it is possible to optimize the tunable functions and find an explicit formula for the optimal functions.\\

{\noindent \textbf{Proof of Part (ii): Fix $\lambda$ (or say $\alpha$) and then optimize $\theta, \beta$.}}\\

We break the proof into two steps. First, we consider the optimization of $\theta$ and $\beta$ for fixed $\jumprate$. 
Then we consider the optimization of $\jumprate$ (namely, optimizing $\alpha$) later. 
Since $\jumprate \ge 0$ and all tunable parameters are state-dependent functions, 
we can simply optimize the rate of change of variance in a pointwise way.

It is thus enough to consider 
\begin{align*}
&\ \min\ \jumprate(\psi) \Big(\underbrace{\inner{b(\psi)}{O b(\psi)} 
+ \frac{2}{\sqrt{\jumprate(\psi)}} \big(\re{\eta(\psi)}\inneravg{\obs}{\psi} + \re{Z(\psi)}\big)}_{=:\mathscr{T}_3(\psi)}\Big)^2.
\end{align*}
The term inside is
\begin{align*}
\mathscr{T}_3(\psi) =\ &\ \innerbig{b(\psi)}{\obs b(\psi)} 
+ \frac{2}{\sqrt{\jumprate(\psi)}} \Big(\re{\eta(\psi)}\inneravg{\obs}{\psi} + \re{Z(\psi)}\Big) \\
\myeq{\eqref{eqn::bOb}}&\  \Big(\frac{\re{\eta(\psi)}}{\sqrt{\jumprate(\psi)}}\Big)^2 \inneravg{\obs}{\psi} 
+ \Big(\frac{\im{\eta(\psi)}}{\sqrt{\jumprate(\psi)}}\Big)^2 \inneravg{\obs}{\psi} 
 + 2 \frac{\re{\eta(\psi)}}{\sqrt{\jumprate(\psi)}} \frac{\re{Z(\psi)}}{\sqrt{\jumprate(\psi)}} 
 + 2\frac{\im{\eta(\psi)}}{\sqrt{\jumprate(\psi)}} \frac{\im{Z(\psi)}}{\sqrt{\jumprate(\psi)}} \\
&\qquad + \frac{\inneravg{L^\dagger O L}{\psi}}{\jumprate(\psi)} 
+ 2 \frac{\re{\eta(\psi)}}{\sqrt{\jumprate(\psi)}} \inneravg{\obs}{\psi} 
+ \frac{2\re{Z(\psi)}}{\sqrt{\jumprate(\psi)}} \\
=\ &\ \Big(\frac{\re{\eta(\psi)}}{\sqrt{\jumprate(\psi)}}+1\Big)^2 \inneravg{\obs}{\psi} 
+ \Big(\frac{\im{\eta(\psi)}}{\sqrt{\jumprate(\psi)}}\Big)^2 \inneravg{\obs}{\psi} 
+ 2 \Big(\frac{\re{\eta(\psi)}}{\sqrt{\jumprate(\psi)}} +1\Big) \frac{\re{Z(\psi)}}{\sqrt{\jumprate(\psi)}} \\
&\qquad + 2\frac{\im{\eta(\psi)}}{\sqrt{\jumprate(\psi)}} \frac{\im{Z(\psi)}}{\sqrt{\jumprate(\psi)}} 
+ \frac{\inneravg{L^\dagger O L}{\psi}}{\jumprate(\psi)} - \inneravg{\obs}{\psi} \\
=\ &\ \Big\lvert \frac{{\eta(\psi)}}{\sqrt{\jumprate(\psi)}}+1 \Big\rvert^2 \inneravg{\obs}{\psi} 
+ 2 \reBig{\big(\frac{\eta(\psi)}{\sqrt{\jumprate(\psi)}} +1\big)^* \frac{{Z(\psi)}}{\sqrt{\jumprate(\psi)}}} 
+ \frac{\inneravg{L^\dagger O L}{\psi}}{\jumprate(\psi)} - \inneravg{\obs}{\psi}.
\end{align*}
By plugging the expression of $\eta$ \eqref{eqn::opti_jump}, we have 
\begin{align*}
&\ \frac{\eta(\psi)}{\sqrt{\jumprate(\psi)}} + 1 \myeq{\eqref{eqn::opti_jump}}\ 
e^{i \beta(\psi)}\underbrace{\sqrt{1 + \frac{\abs{\inneravg{L}{\psi}}^2 
- \inneravg{L^\dagger L}{\psi}}{\jumprate(\psi)}}}_{=: \cstA(\psi)}
- \frac{e^{i\theta(\psi)}}{\sqrt{\jumprate(\psi)}} \inneravg{L}{\psi}.
\end{align*}
Hence, by plugging the last equation and \eqref{eqn::Z}, 
\begin{align*}
\mathscr{T}_3(\psi) =&\ \big\lvert e^{i\beta(\psi)} \cstA(\psi) 
- \frac{e^{i\theta(\psi)}}{\sqrt{\jumprate(\psi)}} \inneravg{L}{\psi} \big\rvert^2 \inneravg{\obs}{\psi} 
+ 2 \re{(e^{-i\beta(\psi)} \cstA(\psi) - \frac{e^{-i\theta(\psi)}}{\sqrt{\jumprate(\psi)}} 
\inneravg{L}{\psi}^*) \frac{e^{i \theta(\psi)} 
\inner{\obs \psi}{L\psi}}{\sqrt{\jumprate(\psi)}}}  \\
&\qquad\qquad + \frac{\inneravg{L^\dagger O L}{\psi}}{\jumprate(\psi)}
- \inneravg{\obs}{\psi} \\
=&\ \big\lvert \cstA(\psi) - \frac{e^{i(\theta(\psi)-\beta(\psi))}}{\sqrt{\jumprate(\psi)}} \inneravg{L}{\psi} \big\rvert^2 \inneravg{\obs}{\psi} 
+ 2 \re{\cstA(\psi) \frac{e^{i (\theta(\psi)-\beta(\psi))} \inner{\obs \psi}{L\psi}}{\sqrt{\jumprate(\psi)}}} \\
&\qquad - 2\re{ \frac{\inneravg{L}{\psi}^*}{\sqrt{\jumprate(\psi)}}\frac{\inner{\obs \psi}{L\psi}}{\sqrt{\jumprate(\psi)}} } 
+ \frac{\inneravg{L^\dagger O L}{\psi}}{\jumprate(\psi)} - \inneravg{\obs}{\psi}.
\end{align*}
Therefore, it is clear that $\mathscr{T}_3$ only depends on the phase difference of $\theta - \beta$. 
Without loss of generality, let us pick $\beta = 0$.
With some simplifications,
\begin{align}
    \label{eqn::T3}
    \begin{aligned}
    \mathscr{T}_3(\psi) =&\ \big\lvert \cstA(\psi) - e^{i \theta(\psi)} \frac{\inneravg{L}{\psi} }{\sqrt{\jumprate(\psi)}} 
\big\rvert^2 \inneravg{\obs}{\psi} + 2 \cstA(\psi)\ \re{\frac{e^{i \theta(\psi)} 
\inner{\obs \psi}{L\psi}}{\sqrt{\jumprate(\psi)}}} \\
&\qquad + \frac{\inneravg{L^\dagger O L}{\psi} - 2 \re{{\inneravg{L}{\psi}^*} \inner{\obs \psi}{L\psi}}}{\jumprate(\psi)}  - \inneravg{O}{\psi},
    \end{aligned}
\end{align}
and recall that we want to minimize 
\begin{align}
\label{eqn::min_Lambda_T3}
\ \min_{\alpha} \min_{\theta}\ \jumprate(\psi) \big(\mathscr{T}_3(\psi)\big)^2.
\end{align}

\begin{lemma}
\label{lemma::ABC}
Suppose $\CA, \CB, \CC$ are real-valued constants, and we want to minimize 
\begin{align}
    \label{eqn::opti_ABC}
\min_{\theta\in [0, 2\pi]}\ \ \abs{\CA \cos(\theta) + \CB \sin(\theta) + \CC}^2.
\end{align}
The optimal $\theta^\star = \Theta(\CA, \CB, \CC)$ where
\begin{align}
    \label{eqn::theta_optimal-proof}
e^{i \Theta(\CA, \CB, \CC)} = 
\left\{\begin{aligned}
& \frac{- \CC \pm \sqrt{\CA^2 + \CB^2 - \CC^2} i}{\CA - i \CB}, \qquad & \text{ if } \sqrt{\CA^2 + \CB^2} \ge \abs{\CC} \text{ and } \CA^2 + \CB^2 > 0; \\
& \frac{{\normalfont \text{sign}}(-\CC) \sqrt{\CA^2 + \CB^2}}{\CA - i \CB}, \qquad & \text{ if } \sqrt{\CA^2 + \CB^2} < \abs{\CC} \text{ and } \CA^2 + \CB^2 > 0; \\
& \text{arbitrary value with norm one}, \qquad & \text{ if } \CA = \CB = 0.
\end{aligned}\right.
\end{align}
The minimum value is
\begin{align}
    \label{eqn::min_ABC}
\Big(\max\{0, \abs{\CC} - \sqrt{\CA^2 + \CB^2}\}\Big)^2.
\end{align}
\end{lemma}

\begin{proof}
    The expression above is $\abs{\re{(\CA-i \CB) e^{i\theta}} + \CC}$. The remaining proof is elementary, and is skipped.
\end{proof}

For the above inner optimization problem in \eqref{eqn::min_Lambda_T3}, namely, $\min_\theta \big(\mathcal{T}_3(\psi)\big)^2$, 
we can organize the expression of \eqref{eqn::T3} in terms of the above optimization problem \eqref{eqn::opti_ABC}
where
coefficients are 
\begin{align}
\label{eqn::ABC}
\left\{
\begin{aligned}
\CA(\psi) = &\  
    \frac{2 \cstA(\psi)}{\sqrt{\jumprate(\psi)}} \Big(- {\re{\inneravg{L}{\psi}}} \inneravg{O}{\psi} + \re{\inneravg{OL}{\psi}} \Big), \\
\CB(\psi) = &\ \frac{2 G(\psi)}{\sqrt{\lambda(\psi)}} \im{\inneravg{L}{\psi}} \inneravg{O}{\psi} - 2 \frac{\cstA(\psi)}{\sqrt{\jumprate(\psi)}} \im{\inneravg{\obs L}{\psi}}, \\
\CC(\psi) = &\ \big(\cstA(\psi)^2-1\big) \inneravg{\obs}{\psi} + \frac{\abs{\inneravg{L}{\psi}}^2 \inneravg{\obs}{\psi} 
    + \inneravg{L^\dagger \obs L}{\psi} 
    - 2 \re{{\inneravg{L}{\psi}^*} \inneravg{OL}{\psi}} }{\jumprate(\psi)} \\
=&\ \frac{\abs{\inneravg{L}{\psi}}^2 \inneravg{O}{\psi}  - \inneravg{L^\dagger L}{\psi} \inneravg{O}{\psi} 
    + \abs{\inneravg{L}{\psi}}^2 \inneravg{O}{\psi} 
    + \inneravg{L^\dagger O L}{\psi} 
    - 2 \re{{\inneravg{L}{\psi}^*} \inneravg{\obs L}{\psi}} }{\jumprate(\psi)},\\
\end{aligned}\right.
\end{align}
where
\begin{align}
\label{eqn::cstA}
\cstA(\psi) = & \sqrt{1 + \frac{\abs{\inneravg{L}{\psi}}^2 - \inneravg{L^\dagger L}{\psi}}{\jumprate(\psi)}}.
\end{align}

In summary, given any fixed $\lambda$ or say $\alpha$, we know that $\theta^*(\psi) = \Theta(\CA(\psi), \CB(\psi), \CC(\psi))$ 
where $\CA(\psi)$, $\CB(\psi)$ and $\CC(\psi)$ are given in \eqref{eqn::ABC}.

\bigskip

{\noindent \textbf{Proof of Part (ii): Optimize $\lambda$ or say $\alpha$.}}\\

As seen above in \eqref{eqn::min_ABC}, the optimization problem for $\lambda$ is 
\begin{align}
\label{eqn::min_alpha}
\begin{aligned}
&\ \min_{\alpha} \lambda(\psi) \Big(\max\Big\{0, \abs{\CC(\psi)} - \sqrt{\CA(\psi)^2 + \CB(\psi)^2}\Big\}\Big)^2 \\
=&\ \min_{\alpha} \Big(\max\Big\{0, \sqrt{\lambda(\psi)} \abs{\CC(\psi)} - \sqrt{\lambda(\psi)} \sqrt{\CA(\psi)^2 + \CB(\psi)^2}\Big\}\Big)^2.
\end{aligned}
\end{align}
Hence, we should minimize the following quantity as much as possible:
\begin{align}
\label{eqn::optimize_alpha}
\min_{\alpha}\ \sqrt{\lambda(\psi)} \abs{\CC(\psi)} - \sqrt{\lambda(\psi)} \sqrt{\CA(\psi)^2 + \CB(\psi)^2}.
\end{align}
When this value is non-positive, then the optimal value in \eqref{eqn::min_Lambda_T3} is zero and it reaches the global minimum that diffusion process could also reach.
When this value is strictly positive, then this shows that the jump process is strictly worsen than the diffusion process in terms of variance reduction.
We remark that the global minimizer for the above may not be unique.
However, the one that minimizes the above value must be a global minimizer. Therefore,
it suffices to consider this optimization problem \eqref{eqn::optimize_alpha} to find an optimal $\alpha$ (or equivalently the jump rate $\jumprate$).

By direct simplification, we have
\begin{align*}
\CA(\psi)^2 + \CB(\psi)^2 &= \frac{\cstA(\psi)^2}{\jumprate(\psi)} \CF(\psi), \qquad \abs{\CC} = \frac{\abs{\CC_1}}{\jumprate},
\end{align*}
where the following two quantities are independent of any tunable parameter:
\begin{align*}
\CF(\psi) &:= 4 \big(-{\re{\inneravg{L}{\psi}}} \inneravg{\obs}{\psi} 
+ \re{\inneravg{OL}{\psi}}\big)^2 
+ 4 {(-\im{\inneravg{OL}{\psi}} + \im{\inneravg{L}{\psi}} \inneravg{O}{\psi})^2}, \\
\CC_1(\psi) &:= \abs{\inneravg{L}{\psi}}^2 \inneravg{O}{\psi} - \inneravg{L^\dagger L}{\psi} \inneravg{O}{\psi} 
+  \abs{\inneravg{L}{\psi}}^2 \inneravg{O}{\psi} 
+ \inneravg{L^\dagger O L}{\psi} 
- 2 \re{{\inneravg{L}{\psi}^*}\inneravg{OL}{\psi}}.
\end{align*}
Hence, the above optimization problem is 
\begin{align*}
&\ \min_\alpha \sqrt{\lambda(\psi)} \abs{\CC(\psi)} - \sqrt{\lambda(\psi)} \sqrt{\CA(\psi)^2 + \CB(\psi)^2} 
=\ \min_\alpha \sqrt{\lambda(\psi)} \frac{\abs{\CC_1(\psi)}}{\lambda(\psi)} - \sqrt{\lambda(\psi)} \frac{G(\psi)}{\sqrt{\lambda(\psi)}} \sqrt{\CF(\psi)} \\
=&\ \min_\alpha \frac{\abs{\CC_1(\psi)}}{\sqrt{\lambda(\psi)}} - G(\psi) \sqrt{\CF(\psi)}
\myeq{\eqref{eqn::cstA}}\ \min_{\alpha} \frac{\abs{\CC_1(\psi)} - \sqrt{\CF(\psi) \alpha}}{\sqrt{\alpha + \CL(\psi)}},
\end{align*}
where the jump rate $\jumprate(\psi) = \alpha(\psi) + {\CL}(\psi)$, and ${\CL}(\psi):= \inneravg{L^\dagger L}{\psi} - \abs{\inneravg{L}{\psi}}^2\ge 0$ is also independent of tunable parameter $\alpha$.
We can show that the above function is monotone decreasing on $\alpha\in [0,\infty)$, so that the global minimum with boundary constraint $\jumprate \le \Lambda$ in \eqref{eqn::deri_var_jump} is achieved when 
\begin{align}
\label{eqn::optimal_alpha}
\jumprate^\star = \Lambda, \qquad \longleftrightarrow \qquad \alpha^\star = \Lambda - (\inneravg{L^\dagger L}{\psi} - \abs{\inneravg{L}{\psi}}^2).
\end{align}

If we optimize the above over the whole half-line $\alpha\in [0,\infty)$,
 and let $\tilde{\alpha}$ be
\begin{align}
\begin{aligned}
\frac{\abs{\CC_1(\psi)} - \sqrt{\CF(\psi) \tilde{\alpha}}}{\sqrt{\tilde{\alpha} + \CL(\psi)}} = 0 \qquad \implies \qquad \tilde{\alpha} &= {\frac{\abs{\CC_1(\psi)}^2}{\CF(\psi)}},
\end{aligned}
\end{align}
we can readily know that for any $\alpha \ge \tilde{\alpha}$, the optimization problem \eqref{eqn::min_alpha} achieves value zero. 
However, since the time step $\Delta t$ cannot be chosen to be infinitesimally small in practice, the optimization problem for the uniformly bounded $\jumprate \le \Lambda$ clearly may not yield a minimum value as low as DO-QSD.

\section{DO-QSD for multiple observables}
\label{appx::multiple_observable}

In the above Theorem~\ref{thm::optimal_soln}, it so far only considered the case of a single observable.
However, it is possible to extend the above result to the case of multiple observables.
For example, if we have two observables $\obs_1$ and $\obs_2$, we can simply treat the observables separately 
so that we run a customized DO-QSD for each observable.
When there is a certain correlation between the two observables, e.g., 
the optimal $\theta$ for $\obs_1$ and $\obs_2$ are considerably close,
then it is likely that we can design a scheme that optimizes both observables simultaneously.

\begin{proposition}
Let $O_1,O_2\cdots O_m$ be a collection of different observables, 
and consider the total variance 
\begin{align}
\var\big(a, \{b_k\}_{k=1}^{K}, t\big) = \sum_{j=1}^m \ee \abs{\inner{\psi(t)}{\obs_j\psi(t)}}^2 - \abs{\tr(\obs_j \rho(t))}^2.
\end{align}
The optimization problem is 
\begin{align*}
\min_{a, \{b_k\}} \frac{\ud}{\ud t} \var\big(a, \{b_k\}_{k=1}^{K}, t\big),
\end{align*}
with $a$, $\{b_k\}_{k=1}^K$ in the form of \eqref{eqn::ab_Cn::2} and \eqref{eqn::norm},
admits the following explicit global minimizers $(\eta_k^\star, \theta_k^\star)$,
where their expressions are given by 
\begin{align}
\label{eqn::muti_global_min}
{\normalfont \text{(multi-DO-QSD)}}\qquad\qquad 
\left\{\ 
\begin{aligned}
\eta_k^\star(\psi) =&\ - e^{i \theta_k^\star(\psi)} \inner{\psi}{L_k\psi},\\
e^{i \theta^\star_k(\psi)} =&\ \pm i\ e^{-i \phase_k(\psi)},
\end{aligned}\right.
\end{align}
where for each $1\le k\le K$, $2 \phase_k(\psi)$ is the phase of
\begin{align*}
\sum_{j=1}^{m} R_{k,j}^2(\psi) e^{2 i P_{k,j}(\psi)},
\end{align*}
where for each $1\le j \le m$, we define $R_{k,j}$ and $P_{k,j}$ as 
$\inner{\obs_j \psi}{L_k\psi} - \inner{\psi}{L_k\psi} \inner{\psi}{\obs_j\psi}$ = $R_{k,j}(\psi)e^{iP_{k,j}(\psi)}$.
\end{proposition}

\begin{proof}
Due to the linearity, we only need to consider the form \eqref{eqn::sde_1}. 
Regarding Lemma~\ref{lem::equiv_opti}, 
it is clear that we need to minimize the term $\loss$, where
\begin{align*}
\loss = \sum_{k=1}^{K} \sum_{j=1}^m 4 \ee\ \absBig{\rebig{e^{i \theta_k(\psi)} \big(\inner{\obs_j \psi}{L_k\psi} - \inner{\psi}{L_k\psi} \inner{\psi}{\obs_j\psi})}}^2.
\end{align*}
Note that this optimization can be performed in the pointwise sense.
Then it amounts to consider the following optimization problem with fixed $k$:
\begin{align*}
    \min_{\theta_k} &\ \sum_{j=1}^m\ \absBig{\rebig{e^{i \theta_k} \big(\inner{\obs_j \psi}{L_k\psi} - \inner{\psi}{L_k\psi} \inner{\psi}{\obs_j\psi})}}^2 \\
 =\min_{\theta_k} & \  \sum_{j=1}^m R_{k,j}^2(\psi)\cos^2(\theta_k + P_{k,j}(\psi)) \\
= \min_{\theta_k} &\  \sum_{j = 1}^m\ \frac{R_{k,j}^2(\psi)}{2} + \sum_{j = 1}^m\frac{R_{k,j}^2(\psi)}{2} \cos\big(2\theta_k + 2P_{k,j}(\psi)\big) \\
= \min_{\theta_k} &\  \sum_{j = 1}^m\ \frac{R_{k,j}^2(\psi)}{2} + \frac{1}{2} \rebig{e^{2 i \theta_k} \sum_{j=1}^{m}\ R_{k,j}^2(\psi) e^{2 i P_{k,j}(\psi)} }\\
=& \sum_{j = 1}^m\ \frac{R_{k,j}^2(\psi)}{2} - \frac{1}{2} \absBig{\sum_{j=1}^m R_{k,j}^2(\psi) e^{2 i P_{k,j}(\psi)}}.
\end{align*}
Note that the first term is independent of $\theta_k$. The minimum is clearly attained when $e^{2 i \theta_k(\psi)} = - e^{-2 i P_k(\psi)}$.
\end{proof}

\section{Supplementary details for \secref{sec::numerics}}
\label{sec::ml}

{\noindent \emph{Parameterization:}} For low-dimensional Lindblad equations, we can parameterize the tunable functions in \eqref{eqn::ab_Cn::2} and use the empirical variance as the loss function to obtain an optimized unraveling scheme. In our experiments for \secref{subsec::eg1}, we represent the real-valued functions $\theta_k$ and $h_k$ using fully connected neural networks. Each $\theta_k$ (and similarly $h_k$) takes complex vectors in $\Complex^n$ as input, so the neural network input is $\begin{bsmallmatrix} \re{\psi}\\ \im{\psi} \end{bsmallmatrix}\in \Real^{2n}$, and the output is a single real value. We use ReLU as the activation function, with a network depth of $3$ and layer width of $36$. With $\theta_k$ and $h_k$ parameterized in this way, $\eta_k$ is fully determined by $h_k$ and $\theta_k$ via \eqref{eqn::norm_condition}, so the drift and diffusion terms are fully specified by \eqref{eqn::ab_Cn::2}. As discussed in the main text, the function $\gamma_k$ does not appear to contribute to variance reduction, so we simply set $\gamma_k = 0$ a priori. For the example \eqref{eg::two_decay}, since $K = 2$, we in total trained $2 K = 4$ such real-valued functions, namely, $\theta_1, \theta_2, h_1, h_2$.\\

{\noindent \emph{Loss function and training details:}} The loss function is defined as
\begin{align*}
\sum_{j=0}^{\numtime} \text{Var}\Big(\Big\{\innerbig{\psi_{t_j}^{(k)}}{\obs \psi_{t_j}^{(k)}}\Big\}_{k=1}^{\numsample}\Big),
\end{align*}
where $\psi_{t_j}^{(k)}$ is the wave function at time $t_j = j \Delta t$ for the $k^{\text{th}}$ sample, and the variance is taken over all samples $1\le k\le \numsample$. The unraveling scheme is solved using the Euler-Maruyama method. Since the norm is theoretically preserved for this setup by \eqref{eqn::norm_condition}, we enforce normalization of the wave function directly in the numerical scheme to improve numerical stability.
The time span is $T = 2$, with a time step of $\Delta t = 0.1$ (so $\numtime = 20$), a sample size of $50$, the training epochs of $3000$, and an exponential-scheduled learning rate of $5\times 10^{-4}$ with decay $0.5$. Note that this time step is larger than that used for actual simulations, in order to reduce the computational cost. Our machine learning results confirm that DO-QSD (a training-free scheme) achieves comparable performance even when the empirical variance is optimized in this direct manner; see Figure~\ref{fig::eg1} and Figure~\ref{fig::eg1::var}.

\end{document}